\documentclass[10pt]{article}
\usepackage[]{fullpage}
\usepackage{amsmath, amsthm, amssymb, algorithm,thmtools,algpseudocode,enumerate,caption}
\usepackage{setspace}
\usepackage{subcaption}
\usepackage{perpage}
\usepackage{fancyhdr}
\usepackage{sectsty}
\usepackage[usenames,dvipsnames,svgnames,table]{xcolor}
\definecolor{darkgreen}{rgb}{0.0,0,0.9}

\usepackage[colorlinks=true,pdfpagemode=UseNone,citecolor=OliveGreen,linkcolor=BrickRed,urlcolor=BrickRed,
pagebackref]{hyperref}

\newtheorem{theorem}{Theorem}[section]

\newtheorem{lemma}{Lemma}[section]

\newtheorem{corollary}[theorem]{Corollary}
\theoremstyle{definition}
\newtheorem{definition}{Definition}[section]
\newtheorem{problem}{Problem}[section]

\newtheorem{remark}{Remark}

\usepackage[numbers]{natbib}
\usepackage{multicol}

\usepackage[cmex10]{mathtools}
\usepackage{bbm}

\newcommand{\GG}{{G}}

\newcommand{\OPT}{\mathsf{OPT}}
\newcommand{\GRD}{{\normalfont\text{greedy}}}

\newcommand{\esp}{$\mathrm{ESP}$}

\newcommand{\GGc}{{G}^\circ}

\newcommand{\EEi}{{E}_{\normalfont\text{init}}}

\newcommand{\VV}{{V}}
\newcommand{\TT}{{T}}
\newcommand{\EE}{{E}}
\newcommand{\xx}{\mathbf{x}}

\newcommand{\cc}{\mathbf{c}}

\newcommand{\dd}{\mathbf{d}}

\newcommand{\zero}{\mathbf{0}}

\newcommand{\LL}{\mathbf{L}}
\newcommand{\KK}{\mathbf{C}}

\newcommand{\MM}{\mathbf{M}}
\newcommand{\NN}{\mathbf{N}}

\newcommand{\ES}{\mathbf{S}}

\newcommand{\WW}{\mathbf{W}}

\newcommand{\yy}{\mathbf{y}}

\newcommand{\zz}{\mathbf{z}}

\newcommand{\pp}{\mathbf{p}}
\newcommand{\ppp}{\boldsymbol\pi}
\newcommand{\AAA}{\mathbf{A}}
\newcommand{\Acal}{\mathcal{A}}
\newcommand{\Bcal}{\mathcal{B}}
\newcommand{\aaa}{\mathbf{a}}
\newcommand{\ee}{\mathbf{e}}

\newcommand{\II}{\mathbf{I}}

\newcommand{\CCC}{\mathcal{M}}

\newcommand{\cQQ}{\mathcal{Q}}

\newcommand{\III}{\mathbf{\mathcal{I}}}

\newcommand{\logtree}{\mathsf{log\,tree}_{n,w}}
\newcommand{\tree}{\mathsf{tree}_{n,w}}
\newcommand{\logTG}{\mathsf{logTG}_{n,w}}

\DeclareMathAlphabet\mathbfcal{OMS}{cmsy}{b}{n}

\DeclareMathOperator*{\rank}{rank}

\newcommand{\diag}{\mathop{\mathrm{diag}}}

\DeclareFontFamily{OT1}{pzc}{}
\DeclareFontShape{OT1}{pzc}{m}{it}{<-> s * [1.2] pzcmi7t}{}
\DeclareMathAlphabet{\mathpzc}{OT1}{pzc}{m}{it}

\algrenewcommand\textproc{}%
\makeatletter
\newcommand{\algrule}[1][.2pt]{\par\vskip.5\baselineskip\hrule height #1\par\vskip.5\baselineskip}
\makeatother

\usepackage{enumitem}

\title{\vspace{1cm}Maximizing the Weighted Number of Spanning
  Trees:\\Near-$t$-Optimal Graphs\thanks{Working paper. 
    \url{kasra.mail@gmail.com} -- \url{https://kasra.github.io}}}
\author{Kasra Khosoussi\thanks{Centre for Autonomous Systems (CAS), University
of Technology Sydney.} \and Gaurav S. Sukhatme\thanks{Department of Computer
  Science, University of
Southern California.} \and
Shoudong Huang\footnotemark[2] \and  Gamini
  Dissanayake\footnotemark[2]}
\begin{document}

\maketitle
\thispagestyle{empty}
\vspace{1.5cm}
\begin{abstract}
  Designing well-connected graphs is a fundamental problem that frequently
  arises in various contexts across science and engineering. 
  The weighted number of spanning
  trees, as a connectivity measure, emerges in numerous
  problems and plays a key role in, e.g., network 
  reliability under random edge failure, estimation over networks and
  D-optimal experimental designs. 
  This paper tackles
  the open problem of designing graphs with the maximum weighted number of
  spanning trees under various constraints.   We reveal several new structures, such as the
  log-submodularity of the weighted number of spanning trees in connected
  graphs. We then exploit these structures and design a pair of efficient
  approximation algorithms with performance guarantees and near-optimality
  certificates.
  Our results can be readily applied to a wide
  verity of applications involving graph synthesis and graph sparsification
  scenarios.
\end{abstract}
\clearpage
\thispagestyle{empty}
\tableofcontents
\clearpage

\section{Introduction}
Various graph connectivity measures have been studied and used in
different contexts. Among them are the combinatorial measures, such as
vertex/edge-connectivity, as well as spectral notions, like
algebraic connectivity \cite{Godsil2001}.  As a connectivity measure, the {number of
spanning trees} (sometimes referred to as \emph{graph complexity} or
\emph{tree-connectivity}) stands out in this list since despite its combinatorial
origin, it can also be characterized solely based on the spectrum of
graph Laplacian. It has been shown that tree-connectivity is associated with
D-optimal (determinant-optimal) experimental designs
\cite{gaffke1982d,cheng1981maximizing,bailey2009combinatorics,Pukelsheim1993}.
The number of spanning trees also appears in
the study of \emph{all-terminal network reliability} under (i.i.d.) random edge failure
(defined as the probability of network being connected)
\cite{kelmans1983multiplicative,weichenberg2004high}.  In particular, it has
been proved that for a given number of edges and vertices, the
\emph{uniformly-most reliable} network, \emph{upon existence}, must have the maximum
number of spanning trees
\cite{bauer1987validity,myrvold1996reliable,boesch2009survey}.
The graph with the maximum number of spanning trees among a finite set of graphs
(e.g., graphs with $n$ vertices and $m$ edges) is called
\emph{$t$-optimal}.
The problem of identifying $t$-optimal graphs under a $(n,m)$ constraint
remains open and has been solved {only} for specific pairs of $(n,m)$; see,
e.g.,
\cite{shier1974maximizing,cheng1981maximizing,kelmans1996graphs,petingi2002new}.
We prove that the (weighted) number of spanning trees in connected graphs can be
posed as a monotone log-submodular function.
This structure enables us to design a complementary greedy-convex pair of approximate algorithms
to synthesize  near-$t$-optimal graphs under several constraints with approximation
guarantees and near-optimality certificates.
\subsection*{Notation}
Throughout this paper, bold lower-case and upper-case letters are reserved for
real vectors and matrices, respectively.
The standard basis for $\mathbb{R}^{n}$ is denoted by
$\{\ee_{i}^n\}_{i=1}^{n}$, and $\ee_{0}^n$ is defined to be the zero $n$-vector.
For any $n \in \mathbb{N}$, $[n]$ denotes the set $\mathbb{N}_{\leq n} = \{1,2,\dots,n\}$. Sets are shown by upper-case letters.
$|\mathcal{X}|$ denotes the cardinality of set $\mathcal{X}$. 
For any finite set $\mathcal{W}$, $\binom{\mathcal{W}}{k}$ is the set of all
$k$-subsets of $\mathcal{W}$.
The eigenvalues of symmetric matrix $\MM$ are denoted by
\mbox{$\lambda_1(\MM) \leq \dots \leq \lambda_n(\MM)$}. $\mathbf{1}$, $\II$ and
$\mathbf{0}$ denote the vector of all ones, identity and zero matrix with appropriate sizes,
respectively.  $\ES_1\succ\ES_2$ means $\ES_1 - \ES_2$ is positive-definite.
The Euclidean norm is denoted by
$\|\cdot\|$.
$\diag(\WW_i)_{i=1}^{k}$ is
the block-diagonal matrix with matrices $(\WW_i)_{i=1}^{k}$ as blocks on its
main diagonal. 
For any graph $\GG$, $\EE(\GG)$ denotes the edge set of $\GG$.
Finally, $\mathbb{S}^{n}_{\geq 0}$ and
$\mathbb{S}^{n}_{> 0}$ denote the set of symmetric positive semidefinite
and symmetric positive definite matrices in $\mathbb{R}^{n\times n}$, respectively.

\section{Background}
\subsection{Preliminaries}
Let $\GG = (\VV,\EE)$ be an undirected graph over
$\VV = [n]$ and with $|\EE| = m$ edges.
By assigning a positive weight to each edge of the graph through $w : \EE \to
\mathbb{R}_{> 0}$, we obtain
$\GG^w = (\VV,\EE,w)$. To shorten our notation let us define 
$w_{uv} \triangleq w(u,v) = w(v,u)$. As it will become clear shortly,
without loss of generality we can assume
$\GG$ is a simple graph since (i) loops do not affect the number of spanning trees,
and (ii) parallel edges can be replaced by a
single edge whose weight is the sum of the weights of the parallel edges.
\mbox{$\WW \triangleq
  \diag\left(w(e_1),\dots,w(e_m)\right)$} denotes the weight
  matrix in which $e_i \in \EE$ is the $i$th edge.
The degree of vertex \mbox{$v \in \VV$}
in ${\GG}$ is denoted by $\deg(v)$. 
Let $\tilde{\AAA}$ be the incidence matrix of $\GG$ after assigning arbitrary orientations
to its edges. The Laplacian matrix of $\GG$ is defined as $\tilde{\LL}
\triangleq \tilde{\AAA}\tilde{\AAA}^\top$.
For an arbitrary choice of $v_0 \in \VV$, let \mbox{${\AAA} \in
\{-1,0,1\}^{(n-1)
\times m}$} be the matrix obtained by removing the row that corresponds
to $v_0$ from $\tilde{\AAA}$. We call $\AAA$ the \emph{reduced incidence matrix}
of $\GG$ after \emph{anchoring} $v_0$. The \emph{reduced Laplacian matrix} of
$\GG$ is defined as $\LL \triangleq \AAA \AAA^\top$. $\LL$ is also known as 
the \emph{Dirichlet} or \emph{grounded Laplacian matrix} of $\GG$.
Note that $\LL$ can also be obtained by removing the row and column associated
to the anchor from the graph Laplacian matrix.
$\AAA$ is full
column rank and consequently $\LL$ is positive definite, iff ${\GG}$ is connected.
For weighted graphs,  \mbox{$\AAA \WW
\AAA^\top$} is the \emph{reduced weighted Laplacian} of $\GG^w$.
 Note that this is a natural generalization of $\LL$, and will reduce to
 its unweighted counterpart if all weights are equal to one (i.e., $\WW = \II$).
  The reduced (weighted) Laplacian matrix can
  be decomposed into the (weighted) sum of \emph{elementary reduced Laplacian
  matrices}:
  \begin{align}
    \LL = \sum_{\mathclap{\{u,v\} \in \EE}} w(u,v) \LL_{uv}
  \end{align} in which $\LL_{uv}\triangleq
  \aaa_{uv} \aaa_{uv}^\top$ and
  $\aaa_{uv} =\ee_u - \ee_v$ is the corresponding column of $\AAA$.
\subsection{Matrix-Tree Theorems}
\label{sec:mainA}
The spanning trees of ${\GG}$ are spanning subgraphs of ${\GG}$ that are also trees. Let $\mathpzc{T}_{\GG}$ and $t(\GG)\triangleq
|\mathpzc{T}_{\GG}|$ denote the
set of all spanning trees of $\GG$ and its number of spanning trees,
respectively.
Let $T_n$ and $K_n$ be, respectively, an arbitrary tree and the complete graph with $n$
vertices. The following statements hold.
\begin{enumerate}
  \item $t(\GG) \geq 0$, and $t(\GG) = 0$ iff $\GG$ is disconnected,
  \item $t(T_n) = 1$,
  \item $t(K_n) = n^{n-2}$ (Cayley's formula),
  \item if $\GG$ is connected, then $t(T_n) \leq t(\GG) \leq t(K_n)$,
  \item if $\GG_1$ is a spanning subgraph of $\GG_2$, then $t(\GG_1) \leq t(\GG_2)$.
\end{enumerate}
Therefore $t(\GG)$ is a sensible measure of
graph connectivity.
The following theorem by Kirchhoff provides an expression for computing $t(\GG)$.

\begin{theorem}[Matrix-Tree Theorem \cite{Godsil2001}]
  Let $\LL_\GG$ and $\tilde{\LL}_\GG$ be, respectively, the reduced Laplacian and the Laplacian matrix of any simple undirected graph
  $\GG$ after anchoring an arbitrary vertex out of its $n$ vertices. The following statements hold.
  \begin{enumerate}
    \normalfont
    \item $t(\GG) = \det(\LL_\GG)$,
    \item $t(\GG) = \frac{1}{n} \prod_{i=2}^{n}
      \lambda_i(\tilde{\LL}_\GG)$.\footnote{Recall that the Laplacian matrix of any connected graph has a zero eigenvalue
	with multiplicity one (see, e.g., \cite{Godsil2001}).}
  \end{enumerate}
  \label{th:max-tree}
\end{theorem}
The matrix-tree theorem can be naturally generalized to weighted graphs, where
each spanning tree is ``counted'' according to its \emph{value} defined below. 
\begin{definition}
  Suppose $\GG = (\VV,\EE,w)$ is a weighted graph with a non-negative weight
  function.  The value of each spanning tree of $\GG$ is measured by the
  following function,
  \begin{align}
   \mathbb{V}_w : \mathpzc{T}_\GG & \to
  \mathbb{R}_{\geq 0} \\
  T & \mapsto \displaystyle\prod_{\mathclap{e \in \, \EE(T)}}
  w(e).
    \label{}
  \end{align}
  Furthermore, we define the weighted number of trees as $t_{w}(\GG) \triangleq
  \sum_{{\TT \in \mathpzc{T}_\GG}} \mathbb{V}_w(\TT)$.
  \label{th:valDef}
\end{definition}
\begin{theorem}[Weighted Matrix-Tree Theorem \cite{Mesbahi2010}]
  For every simple weighted graph \mbox{$\GG = (\VV,\EE,w)$} with $w: \EE \to
  \mathbb{R}_{>0}$ we have $t_w(\GG) = \det \AAA \WW \AAA^\top$. 
  \label{th:wMT}
\end{theorem}
Note that Theorem~\ref{th:wMT}
reduces to Theorem~\ref{th:max-tree} if $w(e) = 1$ for all $e \in \EE$.
Therefore, in the rest of this paper we focus
our attention mainly on weighted graphs.
\begin{definition}
  The weighted \emph{tree-connectivity} of graph $\GG$ is formally defined as
  \begin{equation}
    \tau_w(\GG) \triangleq 
    \begin{cases}
      \log t_w(\GG) & \text{if $t_w(\GG) > 0$,} \\
      0 & \text{otherwise.}
    \end{cases}
    \label{eq:treeConnectivity}
  \end{equation}
\end{definition}
\section{Tree-Connectivity}
\label{sec:main}
\begin{definition}
  \label{def:rand}
  Consider an arbitrary simple undirected graph $\GG^\circ$.
  Let $p_i$ be the probability assigned
  to the $i$th edge, and $\pp$ be
  the stacked vector of probabilities. $\GG \sim \mathbb{G}(\GG^\circ,\pp)$ indicates that
  \begin{enumerate}
    \item $\GG$ is a
  spanning subgraph of $\GGc$.
    \item The $i$th edge of $\GG^\circ$ appears in $\GG$ with
  probability $p_i$, independent of other edges.
  \end{enumerate}
\end{definition}
The naive procedure for computing the expected weighted number of spanning trees
in such random graphs involves a summation over
exponentially many terms.  Theorem~\ref{th:expected} offers an efficient and
intuitive way of computing this expectation in terms of $\GG^\circ$ and $\pp$.
\begin{theorem} 
  For any $\mathbb{G}(\GGc,\pp)$ and $w : \EE(K_n) \to \mathbb{R}_{> 0}$,
    \begin{equation}
      {\mathbb{E}}_{{\GG \sim \mathbb{G}(\GG^\circ,\pp)}}
      \big[t_{w}(\GG)\big] = t_{w_p}(\GG^\circ),
      \label{}
    \end{equation}
 where $w_p(e_i) \triangleq p_i w(e_i)$ for all $e_i \in \EE(\GGc)$.
  \label{th:expected}
\end{theorem}
Note that this expectation can now be computed in $\mathcal{O}(n^3)$ time for
general $\GG^\circ$.
\begin{lemma}
  Let $\GG^{+}$ be the graph obtained by adding $\{u,v\} \notin \EE$ with weight
  $w_{uv}$ to $\GG = (\VV,\EE,w)$.
  Let $\LL_\GG$ be the reduced Laplacian matrix 
  and $\aaa_{uv}$ be the corresponding column of
  the reduced incidence matrix of $\GG$ after anchoring an arbitrary vertex. If
  $\GG$ is connected,
  \begin{equation}
    \tau_{w}(\GG^+) = \tau_{w}(\GG) + \log(1+w_{uv}\Delta^{\GG}_{uv}),
    \label{eq:add}
  \end{equation}
  where $\Delta^{\GG}_{uv} \triangleq \aaa_{uv}^{\top}\LL^{-1}_\GG \aaa_{uv}$.
  \label{th:add}
\end{lemma}
\begin{lemma}
  Similar to Lemma~\ref{th:add}, let $\GG^{-}$ be the graph obtained by removing $\{p,q\} \in \EE$ with
  weight $w_{pq}$ from
  $\EE$. If $\GG$ is connected,
  \begin{equation}
  \tau_{w}(\GG^-) = \tau_{w}(\GG) + \log(1-w_{pq}\Delta^{\GG}_{pq}).
    \label{}
  \end{equation}
  \label{th:del}
\end{lemma}

\begin{corollary}
   Define
    $\mathpzc{T}_\GG^{uv} \triangleq \Big\{T \in \mathpzc{T}_\GG :  \{u,v\} \in
  \EE(T) \Big\}$.
  Then we have
  \begin{equation}
    \Delta_{uv}^{\GG} = {|\mathpzc{T}_{\GG}^{uv}|}/{|\mathpzc{T}_{\GG}|} =
    {|\mathpzc{T}_{\GG}^{uv}|}/{t(\GG)}.
    \label{}
  \end{equation}
  Similarly, for weighted graphs we have
  \begin{equation}
    w_{uv}\Delta_{uv}^{\GG} = \frac{\sum_{T \in \mathpzc{T}_{\GG}^{uv}}
    \mathbb{V}_w(T)}{\sum_{T \in \mathpzc{T}_{\GG}^{\phantom{uv}}} \mathbb{V}_w(T)} = \frac{\sum_{T \in \mathpzc{T}_{\GG}^{uv}}
    \mathbb{V}_w(T)}{t_w(\GG)}.
    \label{}
  \end{equation}
\end{corollary}
Lemmas~\ref{th:add} and \ref{th:del} imply that $w_{uv} \Delta_{uv}^{\GG}$
determines the change in tree-connectivity after adding or removing an edge.
This
term is known as the \emph{effective resistance} between $u$ and $v$. If $\GG$
is an electrical circuit where each edge represents a resistor with
a conductance equal to its weight, then $w_{uv}\Delta_{uv}^{\GG}$ is equal to
the electrical resistance across $u$ and $v$. The effective resistance 
also emerges as a key factor in various other contexts; see, e.g.,
\cite{ghosh2008minimizing,Barooah2007,lovasz1993random}.
Note that although we derived $\Delta_{uv}^{\GG}$ using the
  \emph{reduced} graph Laplacian, it is more common to define the effective
  resistance using the pseudoinverse of graph
  Laplacian $\tilde{\LL}_{\GG}$ \cite{ghosh2008minimizing}. 

Now, on a seemingly unrelated note, we turn our attention to structures
associated to tree-connectivity when seen as a set function.

\begin{definition}
  Let $\VV$ be a set of $n \geq 2$ vertices.
  Denote by $\GG_{\EE}$ the graph $(\VV,\EE)$ for any $\EE \in \EE(K_n)$.
  For any $w : 2^{\EE(K_n)}
  \to \mathbb{R}_{>0}$ define
  \begin{align}
    \mathsf{tree}_{n,w} : 2^{\EE(K_n)} & \to
    \mathbb{R}_{\geq 0} \nonumber \\
    \EE & \mapsto t_w(\GG_{\EE}),
    \label{eq:origTreeSub} \\[0.2cm]
    \mathsf{log\,tree}_{n,w} : 2^{\EE(K_n)} & \to
    \mathbb{R} \nonumber \\
    \EE & \mapsto \tau_w(\GG_{\EE}).
    \label{eq:origTreeSubSub}
  \end{align}
\end{definition}

\begin{definition}[Tree-Connectivity Gain]
  Suppose a \emph{connected} base graph $(\VV,\EEi)$ with $n \geq 2$ vertices
  and an arbitrary positive weight function $w: \EE(K_n)
  \to \mathbb{R}_{>0}$ are given. Define
  \begin{align}
    \logTG : 2^{\EE(K_n)} & \to \mathbb{R}_{\geq 0} \nonumber \\
    \EE & \mapsto \logtree({\EE \cup \EEi})  - \logtree({\EEi}).
    \label{}
  \end{align}
\end{definition}
\begin{definition} Suppose $\mathcal{W}$ is a finite set.
  For any $\xi : 2^{\mathcal{W}} \to \mathbb{R}$,
  \begin{enumerate}
    \item $\xi$ is called \emph{normalized} iff $\xi(\varnothing) = 0$.
    \item $\xi$ is called \emph{monotone} if $\xi(\Bcal) \geq \xi(\Acal)$ for every
      $\Acal$ and $\Bcal$ s.t. $\Acal \subseteq \Bcal \subseteq \mathcal{W}$.
    \item $\xi$ is called  \emph{submodular} iff 
      for every
      $\Acal$ and $\Bcal$ s.t. $\Acal \subseteq \Bcal \subseteq \mathcal{W}$ and
      \mbox{$\forall s \in
	\mathcal{W}
      \setminus \Bcal$}
      we have,
      \begin{equation}
	\xi(\Acal \cup \{s\}) - \xi(\Acal) \geq \xi(\Bcal \cup \{s\}) -
	\xi(\Bcal).
	\label{}
      \end{equation}
    \item $\xi$ is called \emph{supermodular} iff $-\xi$ is submodular. 
    \item $\xi$ is called \emph{log-submodular} iff $\xi$ is positive and
      $\log \xi$
      is submodular.
  \end{enumerate}
\end{definition}
\begin{theorem}
  $\tree$
  is normalized, monotone and supermodular.
  \label{th:treeSupermodular}
\end{theorem}

\begin{theorem}
  $\logTG$ is normalized, monotone and
  submodular.
  \label{th:logTG}
\end{theorem}
Corollary~\ref{cor:expected} follows directly from
Theorems~\ref{th:expected}, \ref{th:treeSupermodular} and \ref{th:logTG}.
\begin{corollary}
  \label{cor:expected}
  The expected weighted number of spanning trees 
   in random graphs is normalized, monotone and supermodular
   when seen as a set function similar to $\tree$. Moreover, the expected
   weighted number of spanning trees can be posed as a log-submodular function
   similar to $\logTG$.
\end{corollary}

\section{\esp{}: Edge Selection Problem}
\label{sec:esp}
\subsection{Problem Definition}
Suppose a connected base graph is given. The edge selection problem (\esp{}) is
a combinatorial optimization problem whose goal is to pick the optimal $k$-set
of edges from a given candidate set of new edges such that the weighted number
of spanning trees after adding those edges to the base graph is maximized.
\begin{problem}[\esp{}]
  Let $\GG_\text{init} = (\VV,\EEi,w)$ be a given connected graph where $w : \EE(K_n) \to
  \mathbb{R}_{>0}$.
  Consider the following scenarios.
  \begin{enumerate}
    \item $k$-\esp{}$^+$: For some $\CCC^+ \subseteq \EE(K_n) \setminus
      \EEi$,
    \begin{equation}
      \begin{aligned}
	& \underset{\EE \subseteq \CCC^+}{\text{maximize}}
	& & t_w(\GG_{\EEi \cup \EE})\\
	& \text{subject to}
	&& |\EE| = k.
      \end{aligned}
      \label{eq:addEdge}
    \end{equation}
    \item $k$-\esp{}$^-$: For some $\CCC^- \subseteq \EEi$,
    \begin{equation}
      \begin{aligned}
	& \underset{\EE \subseteq \CCC^-}{\text{maximize}}
	& & t_w(\GG_{\EEi \setminus \EE})\\
	& \text{subject to}
	&& |\EE| = k.
      \end{aligned}
      \label{eq:delEdge}
    \end{equation}
  \end{enumerate}
\end{problem}
\begin{remark}
  It is easy to see that every instance of $k$-\esp{}$^-$ can be expressed as
  an instance of $d$-\esp{}$^+$ problem for a different base graph, some $d$ and
  a candidate set $\CCC^+$ 
  (and
  vice versa). 
  \label{th:venn}
\end{remark}
\begin{remark}
  The open problem of identifying $t$-optimal graphs
  among all graphs with $n$ vertices and $m$ edges \cite{boesch2009survey}
  is an instance of $k$-\esp{}$^+$ with $k = m$, $\EEi = \varnothing$ and
  $\CCC^+ = \EE(K_n)$.
  \label{th:fullEsp}
\end{remark}
Remarks~\ref{th:venn} and \ref{th:fullEsp} ensure that any algorithm designed
for solving $k$-$\mathrm{ESP}^+$ carries over to the other forms of
$\mathrm{ESP}$.  Therefore, although many graph sparsification and edge pruning
scenarios can be naturally stated as a $k$-\esp{}$^-$, in the rest of this paper
we focus our attention mainly on $k$-\esp{}$^+$.
\subsection{Exhaustive Search}
The brute force algorithm for solving $k$-\esp{}$^+$ requires computing the weighted 
tree-connectivity of every $k$-subset of the candidate set.
$t_w(\GG)$ can be computed by performing a Cholesky decomposition on the reduced
weighted Laplacian matrix which requires $\mathcal{O}(n^3)$ time in general.
This time may significantly reduce for sparse graphs. Let $c \triangleq
|\CCC^+|$. For $k = \mathcal{O}(1)$, the time complexity of
the brute force algorithm is $\mathcal{O}(c^k n^3)$. If $c = \mathcal{O}(n^2)$, this
complexity becomes $\mathcal{O}(n^{2k+3})$, which clearly is not scalable beyond
$k \geq 3$. Moreover, for $k = \alpha \cdot c$ ($\alpha < 1$) the time complexity of exhaustive
search becomes exponential in $c$. To address this problem, in the rest of this
section with propose two efficient approximation algorithms with performance
guarantees by exploiting the inherent structures of tree-connectivity.
\subsection{Greedy Algorithm}
For any $n \geq 2$, $w : \EE(K_n) \to \mathbb{R}_{>0}$, connected $(\VV,\EEi)$, and $\CCC^+ \subseteq \EE(K_n)$
define 
\begin{align}
  \varphi : 2^{\CCC^+} & \to \mathbb{R}_{\geq 0} \\
  \EE & \mapsto  \logTG(\EE)
  \label{}
\end{align}
Note that $\varphi$ is essentially $\logTG$ restricted to $\CCC^+$.
Therefore, Corollary~\ref{cor:add} readily follows from Theorem~\ref{th:logTG}.
\begin{corollary}
  $\varphi$ is normalized, monotone and submodular.
  \label{cor:add}
\end{corollary}
Consequently, $k$-\esp{}$^+$ can be expressed as the problem of maximizing a
normalized monotone submodular
function subject to a cardinality constraint, i.e.,
\begin{equation}
  \begin{aligned}
    & \underset{\EE \subseteq \CCC^+}{\text{maximize}}
    & & \varphi(\EE)\\
    & \text{subject to}
    & & |\EE| = k.
  \end{aligned}
  \label{eq:addEdgePhi}
\end{equation}
Maximizing an arbitrary monotone submodular function subject to a cardinality
constraint \emph{can be} NP-hard in general (see e.g., the Maximum Coverage
problem \cite{hochbaum1996approximation}). Therefore it is reasonable to look for
reliable approximation algorithms. 
In this section we study the greedy algorithm
described in Algorithm~\ref{alg:greed}.
Theorem~\ref{th:Nem} guarantees that Algorithm~\ref{alg:greed} is a constant-factor
approximation algorithm for $k$-\esp$^+$ with a factor of $(1-1/e) \approx
0.63$.
\begin{theorem}[\citet{nemhauser1978analysis}]
  The greedy algorithm attains at least $(1-1/e)f^\star$, where
  $f^\star$ is the maximum of any normalized monotone submodular function
  subject to a cardinality constraint.\footnote{A generalized version of Theorem~\ref{th:Nem} \cite{krause2012submodular}
  states  
  that after $\ell \geq k $ steps, the greedy algorithm is guaranteed to achieve
  at least $(1-e^{-\ell/k})f^\star_k$, where $f^\star_k$ is the
  maximum of $f(\Acal)$ subject to $|\Acal| = k$.}
  \label{th:Nem}
\end{theorem}
\begin{algorithm}[!tb]
  \caption{Greedy Edge Selection}
  \label{alg:greed}
  \begin{algorithmic}[1]
    \Function{{GreedyESP}}{$\LL_\text{init},\CCC^+,k$}
      \State $\EE \gets \varnothing$
      \State $\LL \gets \LL_\text{init}$
      \State $\mathbf{C} \gets$\textrm{Cholesky}({$\LL$})
      \While{$ |\EE| < k$}
      \State $e_{uv}^\star$ $\gets$
      \textrm{BestEdge}($\CCC^+ \setminus \EE,\KK$)
	\State $\EE \gets \EE \cup \{e^\star\}$
	\State $\aaa_{uv} \gets \ee_u - \ee_v$
	\State $\LL \gets \LL +
	w(e_{uv}^{\star})\aaa_{uv}\aaa_{uv}^{\top}$
      \State $\KK \gets
      \textrm{CholeskyUpdate}(\mathbf{C},\sqrt{w(e_{uv}^\star)}\aaa_{uv})$\Comment{{Rank-one
      update}}
      \EndWhile
      \State \Return $\EE$
    \EndFunction
    \algrule
    \Function{{BestEdge}}{$\CCC,\KK$}
    \State ${m} \gets 0$\Comment{Maximum value}
    \ForAll{$e \in \CCC$}\Comment{Parallelizable loop} \label{line:par}
    \State $w_e \gets w(e)$
    \State ${\Delta_e} \gets \textrm{Reff}(e,\KK)$
    \If{${w_e\Delta_e} > {m}$}
    \State $e^\star \gets e$
    \State $m \gets w_e\Delta_e$
    \EndIf
    \EndFor
    \State \Return $e^\star$
    \EndFunction
    \algrule
    \Function{\textrm{Reff}}{$e_{uv},\KK$}\Comment{Effective Resistance}
    \State $\aaa_{uv} \gets \ee_u - \ee_v$
    \State {\color{black!60!white}\texttt{//} solve $\KK \mathbf{x}_{uv} =
    \aaa_{uv}$}
    \State $\mathbf{x}_{uv} \gets$ \textrm{ForwardSolver}($\KK,\aaa_{uv}$)\Comment{Lower Triangular}
    \State $\Delta_{uv} \gets \|\xx_{uv}\|^2$
      \State \Return $\Delta_{uv}$
    \EndFunction
  \end{algorithmic}
\end{algorithm}

\begin{remark}
  Recall that $\varphi$ is normalized by $\mathsf{log\,tree}_{n,w}(\EEi)$, and
  therefore reflects the \emph{tree-connectivity gain} achieved by adding
  $k$ new edges to the original graph $(\VV,\EEi,w)$. In order to avoid any
  confusion, from now on we denote the optimum value of \eqref{eq:addEdgePhi} by
  $\mathsf{OPT}^{\varphi}$, and use $\mathsf{OPT}$ to refer to the maximum achievable tree-connectivity
  in $k$-\esp$^+$. Note that,
  \begin{equation}
    \mathsf{OPT}^\varphi = \mathsf{OPT} - \mathsf{log\,tree}_{n,w}(\EEi).
    \label{}
  \end{equation}
  Let $\EE_\text{greedy}$ be the set of edges picked by
  Algorithm~\ref{alg:greed}. Define $\varphi_\text{greedy} \triangleq
  \varphi(\EE_{\text{greedy}})$. Then, according to Theorem~\ref{th:Nem},
  $\varphi_{\text{greedy}} \geq (1-1/e)\,\OPT^\varphi$ and therefore,
  \begin{equation}
    \mathsf{log\,tree}_{n,w}(\EE_{\text{greedy}} \cup \EEi) \geq (1-1/e)\,\OPT +
    1/e\,\mathsf{log\,tree}_{n,w}(\EEi).
    \label{eq:boundNem}
  \end{equation}
\end{remark}

Algorithm~\ref{alg:greed} starts with an empty set of edges, and in each round picks the
edge that maximizes the weighted tree-connectivity of the graph, until the
cardinality requirement is met. Hence now we need a procedure for
finding the
edge that maximizes the weighted tree-connectivity.
An efficient strategy is to use Lemma~\ref{th:add} and
pick the edge with the highest effective resistance $w_{uv} \Delta_{uv}$.  To
compute $\Delta_{uv} = \aaa_{uv}^\top \LL^{-1} \aaa_{uv}$, we first compute the
 Cholesky factor of the reduced weighted Laplacian matrix of the current
graph $\LL = \KK \KK^\top$. Next, we note that $\Delta_{uv} = \|\xx_{uv}\|^2$ where
$\xx_{uv}$ is the solution of the triangular system $\KK \xx_{uv} = \aaa_{uv}$.
$\xx_{uv}$ can be computed by forward substitution in $\mathcal{O}(n^2)$ time. The time
complexity of each round is dominated by the $\mathcal{O}(n^3)$ time required for computing
the Cholesky factor $\KK$.
In the $i$th round, Algorithm~\ref{alg:greed} has to compute $c - i$
effective resistances where $c = |\CCC^+|$. For $k = \alpha \cdot c$ ($\alpha <
1$), evaluating effective resistances takes
$\mathcal{O}(c^2\,n^2)$ time. If $k = \mathcal{O}(1)$, this time reduces to $\mathcal{O}(c\,n^2)$.
Also, note that upon computing the Cholesky
factor once in each round, $\xx_{uv}$'s can be computed in parallel by
solving $\KK \xx_{uv} = \aaa_{uv}$ for different values of $\aaa_{uv}$ (see line~\#\ref{line:par} in
Algorithm~\ref{alg:greed}). 
We can avoid the $\mathcal{O}(k\,n^3)$ time spent on repetitive Cholesky
factorization by factorizing $\LL_\text{init}$ once, followed by
$k-1$ rank-one updates, each of which takes $\mathcal{O}(n^2)$ time.
Therefore, the total time complexity of
Algorithm~\ref{alg:greed} for $k = \mathcal{O}(1)$ and $k = \alpha\cdot c$ will be $\mathcal{O}(n^3 + c\,n^2)$ and $\mathcal{O}(n^3 +
c^2\,n^2)$, respectively. In the worst case of $\CCC^+ = \EE(K_n)$, $c = \mathcal{O}(n^2)$ and therefore we get
$\mathcal{O}(n^4)$ and $\mathcal{O}(n^6)$, respectively, for $k =
\mathcal{O}(1)$ and $k = \alpha \cdot c$.
Finally, note that for sparse graphs this complexity
drops significantly given a sufficiently good fill-reducing permutation for the
reduced weighted graph Laplacian.

\subsection{Convex Relaxation}
Now we take a different approach and design an efficient approximation algorithm
for $k$-\esp{}$^+$ by means of convex relaxation. We begin by assigning an
auxiliary variable $0 \leq \pi_i \leq 1$ to each candidate edge $e_i \in \CCC^+$.  Let $\ppp
\triangleq [\pi_1\,\pi_2\,\dots\,\pi_c]^\top$ be the stacked vector of auxiliary variables
in which $c = |\CCC^+|$.
Let $\GG= (\VV,\EEi,w)$ be the given base graph. Define
\begin{equation}
  \LL(\ppp;\GG,\CCC^+) \triangleq \sum_{\mathclap{e_i \in \EEi}} \LL_{e_i} +
  \sum_{\mathclap{e_j \in \CCC^+}} \pi_{j} \LL_{e_j} = \AAA
  \WW^\prime\hspace{-0.09cm} \AAA^\top,
  \label{eq:Lpi}
\end{equation}
where $\LL_{e_k}$ is the corresponding reduced elementary weighted
Laplacian, $\AAA$ is the reduced incidence matrix of $(\VV,\EEi \cup \CCC^+)$, and $\WW^\prime \triangleq
\diag(w^\prime(e_1),\dots,w^\prime(e_{s}))$ in which $s \triangleq |\EEi| +
|\CCC^+|$ and,
\begin{equation}
  w^\prime(e_i) \triangleq
  \begin{cases}
    \pi_i w(e_i) & e_i \in \CCC^+, \\
    w(e_i) & e_i \notin \CCC^+.
  \end{cases}
  \label{}
\end{equation}

\begin{lemma}
  $\LL(\ppp)$ is positive definite iff
  $(\VV,\EEi \cup \CCC^+)$ is connected.
  \label{lm:connected}
\end{lemma}
Note that every $k$-subset of $\CCC^+$ is optimal for
$k$-\esp{}$^+$ if $(\VV,\EEi \cup
\CCC^+)$ is not connected. Therefore, if we ignore this degenerate case, we can
safely assume that
$\LL(\ppp;\GG,\CCC^+)$ is positive definite. With a slight abuse of
notation, from now
on we drop the parameters from $\LL(\ppp;\GG,\CCC^+)$ and use $\LL(\ppp)$ whenever $\GG$ and $\CCC^+$
are clear from the context. 
Now consider the following optimization problem over $\ppp$.
\begin{equation}
  \begin{aligned}
    & \underset{\ppp}{\text{maximize}}
    & & \log\det {\LL(\ppp)}\\
    & \text{subject to}
    && \|\ppp\|_0 = k,\\
    &&& 0\leq \pi_i \leq {1}, \, \forall i \in [c].
  \end{aligned}
  \label{eq:conv0}
  \tag{P$_1$}
\end{equation}
\ref{eq:conv0} is equivalent to our former definition of
$k$-\esp{}$^+$. The auxiliary variables act as selectors: the $i$th candidate edge is
selected iff $\pi_i = 1$.
The objective function rewards strong weighted tree-connectivity. The
combinatorial difficulty of \esp{} here is embodied in the
non-convex $\ell_0$-norm constraint. 
 It is easy to see that at the optimal solution, auxiliary variables
take binary values.
Therefore \ref{eq:conv0} can also be expressed as
\begin{equation}
  \begin{aligned}
    & \underset{\ppp}{\text{maximize}}
    & & \log\det {\LL(\ppp)}\\
    & \text{subject to}
    && \|\ppp\|_1 = k,\\
    &&& \pi_i \in \{0,1\}, \, \forall i \in [c].
  \end{aligned}
  \label{eq:conv0bin}
  \tag{P$^\prime_1$}
\end{equation}
A natural choice for relaxing \ref{eq:conv0bin} is to replace $\pi_i \in \{0,1\}$
with $0 \leq \pi_i \leq 1$, i.e.,
\begin{equation}
  \begin{aligned}
    & \underset{\ppp}{\text{maximize}}
    & & \log\det {\LL(\ppp)}\\
    & \text{subject to}
    && \|\ppp\|_1 = k,\\
    &&& 0\leq \pi_i \leq {1}, \, \forall i \in [c].
  \end{aligned}
  \label{eq:conv1}
  \tag{P$_2$}
\end{equation}
The feasible set of \ref{eq:conv1} contains that of \ref{eq:conv0} (or,
equivalently, \ref{eq:conv0bin}), and therefore the optimum value of \ref{eq:conv1} is
an upper bound for the optimum of \ref{eq:conv0} (or, equivalently,
\ref{eq:conv0bin}).
Note that the $\ell_1$-norm here is identical to $\sum_{i=1}^{c} \pi_i$.
\ref{eq:conv1} is a convex optimization problem since the objective function
(tree-connectivity) is concave and the constraints are linear and affine in $\ppp$.
In fact, \ref{eq:conv1} is an instance of the $\mathrm{MAXDET}$ problem
\cite{vandenberghe1998determinant} subject to additional affine 
constraints on $\ppp$.
It is worth noting that \ref{eq:conv1} can be reached also by relaxing the non-convex $\ell_0$-norm
constraint in \ref{eq:conv0} by a convex $\ell_1$-norm constraint (essentially $\sum_{i=1}^c \pi_i = k$).
Furthermore, \ref{eq:conv1} is also closely related to a $\ell_1$-regularalized
instance of $\mathrm{MAXDET}$,
\begin{equation}
  \begin{aligned}
    & \underset{\ppp}{\text{maximize}}
    & & \log\det {\LL(\ppp)} - \lambda \, \|\ppp\|_1\\
    & \text{subject to}
    && 0\leq \pi_i \leq {1}, \, \forall i \in [c].
  \end{aligned}
  \label{eq:conv1b}
  \tag{P$_3$}
\end{equation}
This problem is a penalized form of \ref{eq:conv1}; these two problems
are equivalent for some positive value of $\lambda$.
Problem~\ref{eq:conv1b} is also a convex optimization problem for non-negative
$\lambda$.
The $\ell_1$-norm in \ref{eq:conv1b} encourages sparser $\ppp$, while the
log-determinant rewards stronger tree-connectivity.  The penalty coefficient
$\lambda$ is a parameter that specifies the desired degree of sparsity, i.e.,
larger $\lambda$ yields a sparser vector of selectors $\ppp$.  

Problem~\ref{eq:conv1} (and \ref{eq:conv1b}) can be solved efficiently using
interior-point methods \cite{Boyd2004}.  After finding a globally optimal
solution $\ppp^\star$ for the relaxed problem \ref{eq:conv1}, we ultimately need
to map it into a feasible $\ppp$ for \ref{eq:conv0}, i.e., picking $k$ edges
from the candidate set $\CCC^+$. First note that if $\ppp^\star \in \{0,1\}^c$,
it means that $\ppp^\star$ is already an optimal solution for $k$-\esp{}$^+$
and \ref{eq:conv0}.  However, in the more likely case of $\ppp^\star$ containing
fractional values, we need a \emph{rounding procedure} to set $k$ auxiliary
variables to one and others to zero. The most intuitive choice is to pick
the $k$ edges with the largest $\pi^\star_i$'s. Another (approximate) rounding
strategy (and a justification for picking the $k$ largest $\pi^\star_i$) emerges
from interpreting $\pi_i$ as the probability of selecting the $i$th candidate
edge.  Theorem~\ref{th:cvx} provides a new interesting way of interpreting the
convex relaxation of \ref{eq:conv0} by \ref{eq:conv1}.
\begin{theorem}
  \label{th:cvx}
  Define $\EE_\bullet \triangleq \EEi \cup \CCC^+$ and $\GG_{\bullet} \triangleq
  (\VV,\EE_\bullet,w)$. Let $\ppp_\bullet =
  [\pi_1 \,\dots\,\pi_{s}]^\top \in (0,1]^s$ such that $s \triangleq |\EEi| +
  |\CCC^+|$ and $\pi_i = 1$ if $e_i \in \EEi$. Then we have
  \begin{align}
    {\mathbb{E}}_{{\mathcal{H} \sim
    \mathbb{G}(\GG_\bullet,\ppp_\bullet)}} &
    \big[t_w(\mathcal{H})\big]  = \det 
    {\LL(\ppp)} \label{eq:i}, \\
    {\mathbb{E}}_{{\mathcal{H} \sim
  \mathbb{G}(\GG_\bullet,\ppp_\bullet)}}&
  \big[|\EE(\mathcal{H})| -  |\EEi|\big] =
  \sum_{\mathclap{e_i \in \CCC^+}} \pi_i = \|\ppp\|_1.
    \label{eq:ii}
  \end{align}
\end{theorem}
Note that \eqref{eq:i} and \eqref{eq:ii} appear in the objective function and
the constraints of
\ref{eq:conv1}, respectively. Thus \ref{eq:conv1} can be rewritten as
\begin{equation}
  \begin{aligned}
    & \underset{\ppp}{\text{maximize}}
    & & {\mathbb{E}}_{{\mathcal{H} \sim
    \mathbb{G}(\GG_\bullet,\ppp_\bullet)}}
    \big[t_w(\mathcal{H})\big] \\
    & \text{subject to}
    && {\mathbb{E}}_{{\mathcal{H} \sim
  \mathbb{G}(\GG_\bullet,\ppp_\bullet)}}
  \big[|\EE(\mathcal{H})| \big] = k + |\EEi|,\\
    &&& 0\leq \pi_i \leq {1}, \, \forall i \in [s].
  \end{aligned}
  \label{eq:conv1prime}
  \tag{P$_2^\prime$}
\end{equation}
This offers a new narrative: the objective in \ref{eq:conv1} is to find
the optimal probabilities $\ppp^\star$ for sampling edges from $\CCC^+$ such that the
weighted number of spanning trees is maximized in \emph{expectation}, while the
\emph{expected} number of newly selected edges is equal to $k$.
In other words, \ref{eq:conv1} can be seen as a convex relaxation of
\ref{eq:conv0} at the expense of maximizing the objective and satisfying
the
constraint, both \emph{in expectation}.
This new interpretation motivates an approximate randomized
rounding procedure that picks $e_i \in \CCC^+$ with probability $\pi_i^\star$.
According to Theorem~\ref{th:cvx}, this randomized rounding scheme, in average, attains $\det\LL(\ppp^\star)$ by
picking $k$ new edges in average.
\begin{theorem}
For any $0 < \epsilon < 1$
and $\delta > 0$,
  \begin{align}
    \mathbb{P}\,\big[|\EE^\star| < (1-\epsilon)k\big] & <
    \exp\left(-\epsilon^2k/2\right), \\ 
    \mathbb{P}\,\big[|\EE^\star| > (1+\delta)k\big] & <
    \exp\left(-\delta^2k/3\right),
    \label{}
  \end{align}
  where $\EE^\star$ is the set of selected edges by the randomized rounding
  scheme defined above.
  \label{th:Cher}
\end{theorem}
Theorem~\ref{th:Cher} ensures that the probability of the events in which the
aforementioned randomized rounding
strategy picks too many/few edges (compared to $k$) decay exponentially.
Note that this new narrative offers another intuitive justification for deterministically picking the $k$ edges
with largest $\pi_i^\star$'s.
Finally, we believe that Theorems~\ref{th:cvx} and \ref{th:Cher} can potentially
be used as building blocks to design new randomized rounding schemes.

\subsection{Certifying Near-Optimality}
The proposed approximation algorithms also provide \emph{a posteriori} lower and
upper bounds for the maximum achievable tree-connectivity in \esp{}. Let
$\EE_\text{greedy}$, $\EE_\text{cvx}$  be the solutions returned by the greedy
and convex\footnote{Picking the $k$ edges with the largest
  $\pi^\star_i$'s from the solution of \ref{eq:conv1}.} approximation algorithms, respectively. Let $\tau^\star_\text{cvx}$ be the
  optimum value of \ref{eq:conv1} and define $\tau_\text{init} \triangleq
  \mathsf{log\,tree}_{n,w}(\EEi)$, $\tau_\text{cvx} \triangleq \mathsf{log\,tree}_{n,w}(\EE_\text{cvx} \cup
  \EEi)$ and $\tau_\text{greedy}  \triangleq \mathsf{log\,tree}_{n,w}(\EE_\text{greedy} \cup
  \EEi)$. 
\begin{corollary}
  \label{cor:bound}
\begin{equation}
  \normalfont
  \max\,\Big\{ \tau_\text{greedy},\tau_\text{cvx} \Big\}
  \leq
  \OPT
  \leq
  \min\,\Big\{ \zeta\tau_\text{greedy} +
  (1-\zeta)\tau_\text{init},\tau^\star_\text{cvx} \Big\}
  \label{eq:optbound}
\end{equation}
  where $\zeta \triangleq {(1-1/e)}^{-1} \approx 1.58$.\footnote{Furthermore, recall that the
leftmost term in \eqref{eq:optbound} is bounded from below by the expression given in
\eqref{eq:boundNem}.} 
\end{corollary}
Corollary~\ref{cor:bound} can be used as a tool to asses the quality of any
suboptimal design.
Let $\Acal$ be an arbitrary $k$-subset of $\CCC^+$ and $\tau_\Acal =
\mathsf{log\,tree}_{n,w}(\Acal \cup \EEi)$. Define $\mathcal{U} \triangleq 
  \min\,\Big\{ \zeta\tau_\text{greedy} +
(1-\zeta)\tau_\text{init},\tau^\star_\text{cvx} \Big\}$. $\mathcal{U}$ can be
  computed by running the proposed greedy and convex approximation algorithms.
  From Corollary~\ref{cor:bound} it readily follows that $\OPT - \tau_\Acal \leq
  \mathcal{U} - \tau_\Acal$ and $\OPT/\tau_\Acal \leq \mathcal{U}/\tau_\Acal$.
  Therefore, although we may not have direct access to $\OPT$, we can still
  certify the near-optimality of any design such as $\Acal$ whose $\delta \triangleq \mathcal{U} -
  \tau_\Acal$ is sufficiently small.
\subsection{Numerical Results}
\begin{figure}[t]
  \centering
  \begin{subfigure}[t]{0.31\textwidth}
    \includegraphics[width=\textwidth]{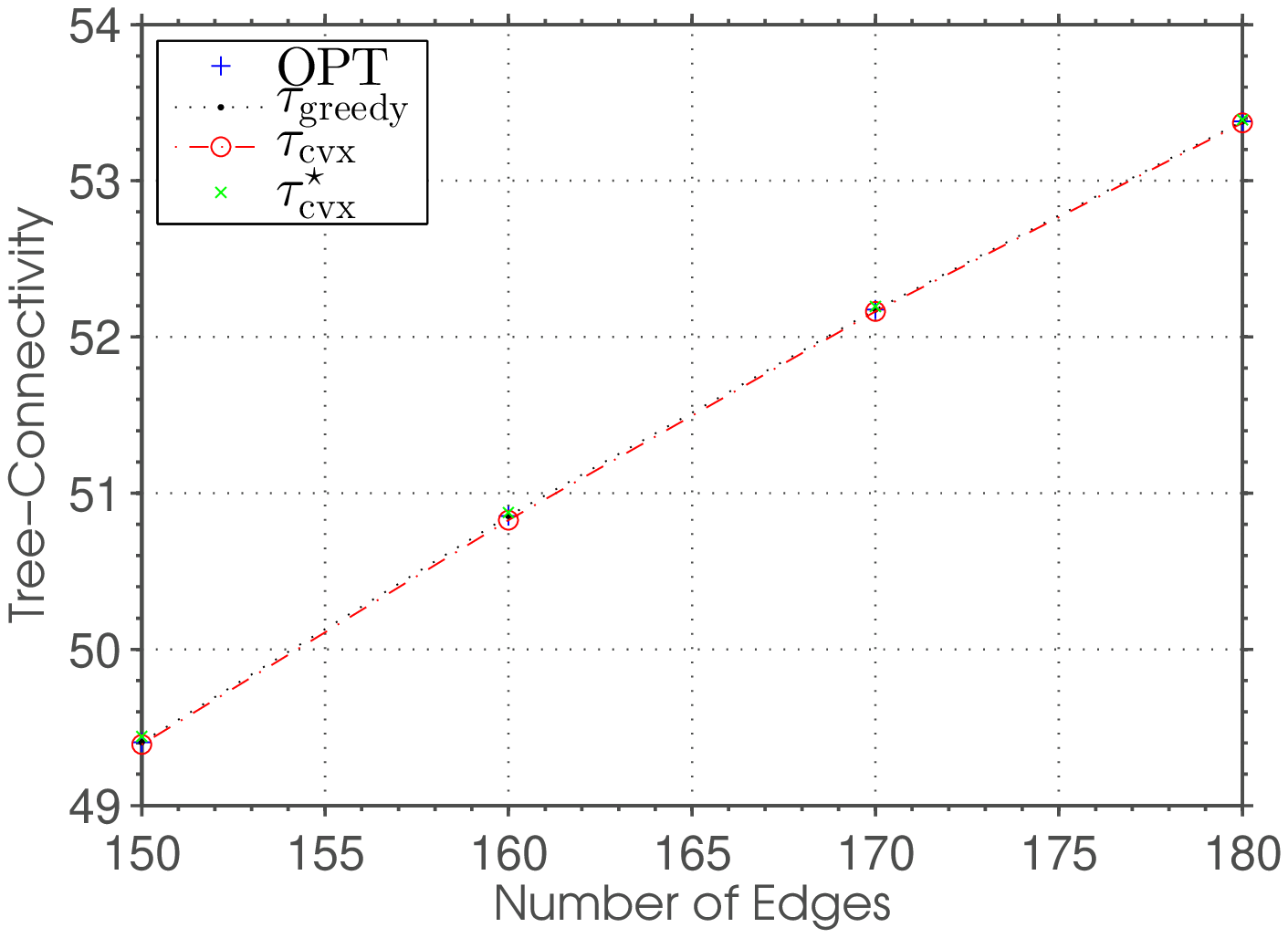}
    \caption{\small Varying $|\EE|$ for $k=5$ and $|\VV| = 20$}
      \label{fig:varEdgesOpt}
    \end{subfigure}
    ~
  \begin{subfigure}[t]{0.31\textwidth}
    \includegraphics[width=\textwidth]{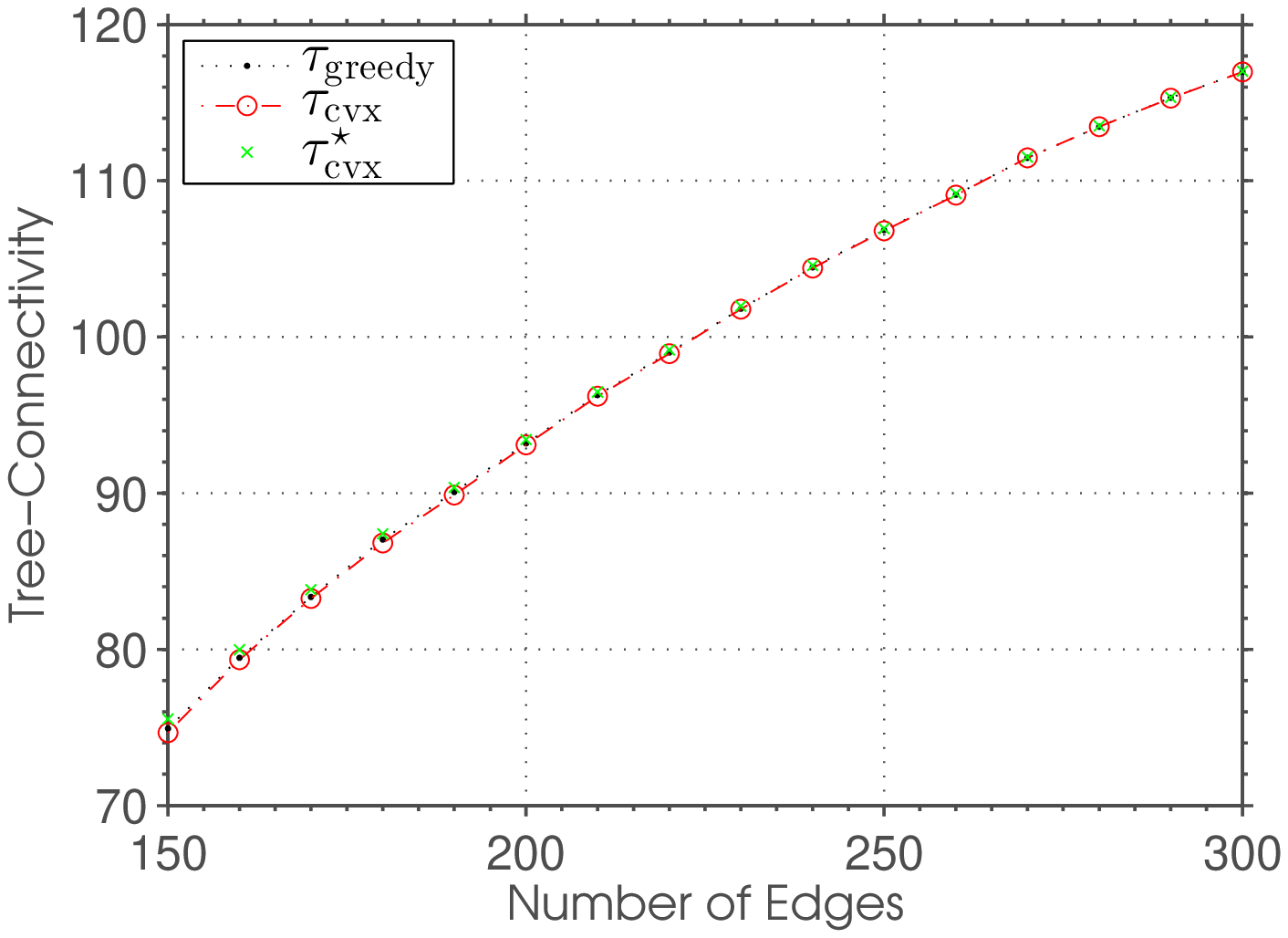}
    \caption{\small Varying $|\EE|$ for $k = 5$ and $|\VV| = 50$}
    \label{fig:varEdges}
  \end{subfigure}
  ~
  \begin{subfigure}[t]{0.31\textwidth}
    \includegraphics[width=\textwidth]{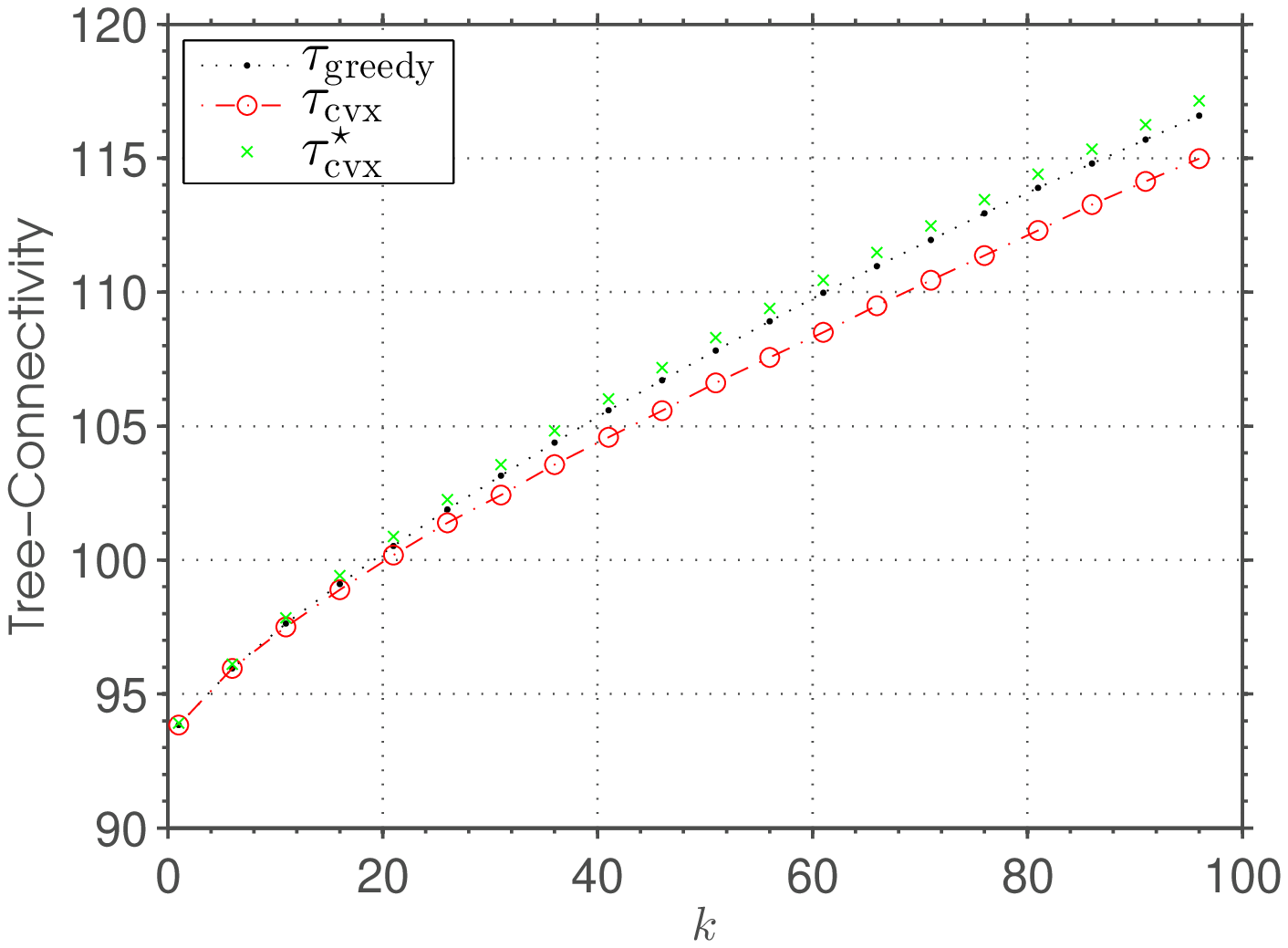}
    \caption{\small Varying $k$ for $|\VV| = 50$ and $|\EE| = 200$}
    \label{fig:varK}
  \end{subfigure}
  \caption{\small \esp{} on randomly generated graphs.
  Recall that according to Corollary~\ref{cor:bound}, $\tau_\text{greedy} \leq \mathsf{OPT} \leq
  \tau^\star_\text{cvx}$.}
  \label{fig:esps}
\end{figure}
We implemented Algorithm~\ref{alg:greed} in MATLAB.
Problem~\ref{eq:conv1}
is modelled using CVX \cite{cvx,gb08} and YALMIP \cite{YALMIP}, and solved
using SDPT$3$ \cite{tutuncu2003solving}.
Figure~\ref{fig:esps} illustrates the performance of our approximate
solutions to $k$-\esp{}$^+$ in randomly generated graphs. 
The search space in these experiments is $\CCC^+ = \EE(K_n) \setminus \EEi$.
Figures~\ref{fig:varEdgesOpt} and \ref{fig:varEdges} show tree-connectivity as a
function of number of randomly generated edges for a fixed $k=5$ and,
respectively, $|\VV| = 20$ and $|\VV| = 50$. 
Our results indicate that both algorithms exhibit
remarkable performances for $k = 5$. Note that computing $\mathsf{OPT}$ by exhaustive search is only feasible in
small instances such as Figure~\ref{fig:varEdgesOpt}. Nevertheless,
computing the exact $\mathsf{OPT}$ is not crucial for evaluating our approximate
algorithms, as it is tightly bounded in
$\mathopen{[}\tau_{\text{greedy}},{\tau}^\star_{\text{cvx}}\mathclose{]}$
as predicted by Corollary~\ref{cor:bound} (i.e., between each black $\mathrm{\cdot}$ and green $\times$).
Figure~\ref{fig:varK} shows the results obtained for varying $k$. The optimality
gap for $\tau_\text{cvx}$ gradually grows as the planning horizon $k$ increases.
Our greedy algorithm, however, still yields a near-optimal approximation. 
\section{Beyond $k$-\esp{}$^+$}
\subsection{Matroid Constraints}
Recall that $\varphi$ is monotone.
Therefore, except the degenerate case of $(\VV,\EEi \cup \CCC^+)$ not being connected, replacing the
cardinality constraint $|\EE| = k$ in $k$-\esp{}$^+$ with an inequality
constraint $|\EE| \leq k$ does not affect the set of optimal solutions. 
Consider the \emph{uniform matroid} 
\cite{oxley2006matroid} 
defined as
$(\CCC^+,\III_\text{U})$ where 
\[\III_\text{U} \triangleq \big\{\Acal \subseteq
\CCC^+ \,:\, |\Acal| \leq k\big\}.\]
The inequality cardinality
constraint can be expressed as $\EE \in \III_\text{U}$. 
\begin{definition}[Partition Matroid]
  Let $\mathcal{M}_1^+,\ldots,\mathcal{M}_\ell^+$ be a partition for 
  $\mathcal{M}^+$. Assign an integer (budget) $0\leq k_i \leq |\CCC^+_i|$ to each
  $\CCC^+_i$. Define \[\mathcal{I}_\text{P} \triangleq
    \Big\{\mathcal{A} \subseteq \mathcal{M}^+ \,:\, |\mathcal{A} \cap
    \mathcal{M}_i^+ | \leq k_i \,\,\text{   for   } i \in [\ell]\Big\}.\] The
    pair $(\mathcal{M}^+,\mathcal{I}_\text{P})$ is called a \emph{partition
    matroid}.
\end{definition}
Now let us consider \esp{} under a partition matroid constraint; i.e.,
\begin{equation}
  \begin{aligned}
    & \underset{}{\text{maximize}}
    & & \varphi(\EE)\\
    & \text{subject to}
    & & \EE \in \III_\text{P}.
  \end{aligned}
  \label{eq:addEdgePM}
\end{equation}
Note that $k$-\esp{}$^+$ is a special case of this problem
with $\ell = 1$ and $k_1 = k$. Now, by choosing different partitions for 
$\CCC^+$ and different budgets $k_i$ we can model a wide variety of
graph synthesis problems. For example consider the following extension of
$k$-\esp{}$^+$,
\begin{equation}
  \begin{aligned}
    & \underset{\EE \subseteq \CCC^+,|\EE|\leq k}{\text{maximize}}
    & & \varphi(\EE)\\
    & \text{subject to}
    & & \deg(v) \leq d.
  \end{aligned}
  \label{eq:addEdgeDeg}
\end{equation}
Define  $\CCC^+_v \triangleq \Big\{ e \in \CCC^+ \, : \, v \in e \Big\}$.
Now note that the constraints in \eqref{eq:addEdgeDeg} can be expressed as a partition
matroid with two blocks: (i) $\CCC^+_v$ with a budget of $k_1 = d$, and (ii) $\CCC^+ \setminus
\CCC^+_v$ with a budget of $k_2 = k-d$.
\subsubsection{Greedy Algorithm}
\begin{theorem}[\citet{fisher1978analysis}]
  The greedy algorithm attains at least $(1/2)f^\star$, where
  $f^\star$ is the maximum of any normalized monotone submodular function
  subject to a matroid constraint.
  \label{th:mat}
\end{theorem}
According to Theorem~\ref{th:mat}, a slightly modified version of
Algorithm~\ref{alg:greed}, that abides by the matroid constraint while
greedily choosing the next best edge, yields a $\frac{1}{2}$-approximation \cite{fisher1978analysis,krause2012submodular}.
\subsubsection{Convex Relaxation}
The proposed convex relaxation of $k$-\esp{}$^+$ can be modified to
handle a partition matroid constraint. First note that \eqref{eq:addEdgePM} can be expressed
as
\begin{equation}
  \begin{aligned}
    & \underset{\ppp}{\text{maximize}}
    & & \log\det {\LL(\ppp)}\\
    & \text{subject to}
    && \sum_{\mathclap{e_i \in \CCC^+_j}} \pi_i \leq k_j, \, \forall j \in [\ell]\\
    &&&\pi_i \in \{0,1\}, \,\,\, \forall i \in [c].
  \end{aligned}
  \label{eq:nconvPM}
  \tag{P$_4$}
\end{equation}
Relaxing the binary constraints on $\pi_i$'s yields
\begin{equation}
  \begin{aligned}
    & \underset{\ppp}{\text{maximize}}
    & & \log\det {\LL(\ppp)}\\
    & \text{subject to}
    && \sum_{\mathclap{e_i \in \CCC^+_j}} \pi_i \leq k_j, \, \forall j \in [\ell]\\
    &&& 0\leq \pi_i \leq {1}, \,\,\, \forall i \in [c].
  \end{aligned}
  \label{eq:convPM}
  \tag{P$_5$}
\end{equation}
\ref{eq:convPM} is a convex optimization problem and, as before, can be solved 
efficiently using interior-point methods. A simple rounding strategy for
the solution of \ref{eq:convPM} is to pick the edges in $\CCC^+_i$ that are
associated to the $k_i$ largest $\pi_j^\star$'s (for $i \in [\ell]$). Moreover, the bounds in
\eqref{eq:optbound} (with $\zeta = 2$) and Theorem~\ref{th:cvx} can also be
readily generalized to handle partition matroid constraints. In particular, the
optimum value of \ref{eq:convPM} gives an upper bound for the optimum value of
\ref{eq:nconvPM}. Also, similar to Theorem~\ref{th:cvx}, \ref{eq:convPM} can be interpreted as maximizing the
expected value of the weighted number of spanning trees such that the expected
number of new edges sampled from $\CCC^+_i$ is at most $k_i$, for $i \in [\ell]$.
\subsection{Dual of $k$-\esp{}$^+$}
The dual of $k$-\esp{}$^+$ aims to identify and select the
minimal set of new edges from a candidate set $\CCC^+$ such that the
resulting tree-connectivity gain is at least $0 \leq \delta \leq
\varphi(\CCC^+)$ for some given $\delta$; i.e.,
\begin{equation}
  \begin{aligned}
    & \underset{\EE \subseteq \CCC^+}{\text{minimize}}
    & & |\EE| \\
    & \text{subject to}
    & & \varphi(\EE) \geq \delta.
  \end{aligned}
  \label{eq:espDual}
\end{equation}
\subsubsection{Greedy Algorithm}
The greedy algorithm for approximating the solution of \eqref{eq:espDual} is
outlined in
Algorithm~\ref{alg:greedD}. The only difference between
Algorithm~\ref{alg:greed} and Algorithm~\ref{alg:greedD} is that the latter
terminates when the $\delta$-bound is achieved (or, alternatively, when
there are no more
edges left in $\CCC^+$, which indicates an empty feasible set).
\citet{wolsey1982analysis} proves several upper bounds for the ratio between the
objective value achieved by the greedy algorithm and the optimum value of the
following class of problems,
\begin{equation}
  \begin{aligned}
     \underset{\Acal \subseteq \mathcal{W}}{\text{minimize}} \,\,\, |\Acal|
    \,\,\, \text{subject to} \,\,\, \phi(\Acal) \geq \phi_0,
  \end{aligned}
  \label{eq:arbit}
\end{equation}
in which $\phi : 2^\mathcal{W} \to \mathbb{R}$ is an arbitrary monotone
submodular function and $\phi_0 \leq \phi(\mathcal{W})$. Note that our problem
\eqref{eq:espDual} is special case of \eqref{eq:arbit}, and therefore (some of)
the bounds proved by \citet[Theorem 1]{wolsey1982analysis} also hold for
Algorithm~\ref{alg:greedD}.

\begin{theorem}[\citet{wolsey1982analysis}]
  Let $k_\OPT$ and $k_\GRD$ be the global minimum of \eqref{eq:espDual}
  and the objective value achieved by Algorithm~\ref{alg:greedD}, respectively.
  Also, let $\tilde{\EE}_{\normalfont\text{greedy}}$ be the set formed by
  Algorithm~\ref{alg:greedD} one step before termination.
  Then $k_\GRD \leq \gamma \, {k_\OPT}$ in which
 \begin{equation}
   \gamma \triangleq 1 + \log \Big(
   \frac{\delta}{\delta - \varphi(\tilde{\EE}_\GRD)} \Big).
   \label{}
 \end{equation}
  \label{th:dualGreedy}
\end{theorem}
The upper bound given above and some of the other bounds in \cite{wolsey1982analysis}
are \emph{a posteriori} in the sense that they can be computed only \emph{after} running
the greedy algorithm. 
\begin{algorithm}[tb]
  \caption{Greedy Dual Edge Selection}
  \label{alg:greedD}
  \begin{algorithmic}[1]
    \Function{{GreedyDualESP}}{$\LL_\text{init},\CCC^+,\delta$}
      \State $\EE \gets \varnothing$
      \State $\LL \gets \LL_\text{init}$
      \State $\mathbf{C} \gets$\textrm{Cholesky}({$\LL$})
      \While{$ (\log\det\LL < \delta) \,\wedge \, (\EE \neq \CCC^+)$}
      \State $e_{uv}^\star$ $\gets$
      \textrm{BestEdge}($\CCC^+ \setminus \EE,\KK$)
	\State $\EE \gets \EE \cup \{e^\star\}$
	\State $\aaa_{uv} \gets \ee_u - \ee_v$
	\State $\LL \gets \LL +
	w(e_{uv}^{\star})\aaa_{uv}\aaa_{uv}^{\top}$
      \State $\KK \gets
      \textrm{CholeskyUpdate}(\mathbf{C},\sqrt{w(e_{uv}^\star)}\aaa_{uv})$\Comment{{Rank-one
      update}}
      \EndWhile
      \State \Return $\EE$
    \EndFunction
  \end{algorithmic}
\end{algorithm}

\subsubsection{Convex Relaxation}
Let $\tau_\text{init} \triangleq \log\det {\LL(\zero)}$. The dual
problem can be expressed as
\begin{equation}
  \begin{aligned}
    & \underset{\ppp}{\text{minimize}}
    && \sum_{i=1}^c \pi_i \\
    & \text{subject to}
    & & \log\det {\LL(\ppp)} \geq \delta + \tau_\text{init},\\
    &&& \pi_i \in \{0,1\}, \, \forall i \in [c].
  \end{aligned}
  \label{eq:dconv0}
  \tag{D$_1$}
\end{equation}
The combinatorial difficulty of the dual formulation of \esp{} is manifested in the
 binary constraints of \ref{eq:dconv0}. Relaxing these constraints into
 $0 \leq \pi_i \leq 1$ yields the following convex optimization problem,
 \begin{equation}
  \begin{aligned}
    & \underset{\ppp}{\text{minimize}}
    && \sum_{i=1}^{c} \pi_i \\
    & \text{subject to}
    & & \log\det {\LL(\ppp)} \geq \delta + \tau_\text{init},\\
    &&& 0\leq \pi_i \leq {1}, \, \forall i \in [c].
  \end{aligned}
  \label{eq:dconv1}
  \tag{D$_2$}
\end{equation}
\ref{eq:dconv1} can be solved efficiently using interior-point methods.
Let $\ppp^\star$ be the minimizer of
\ref{eq:dconv1}. $\sum_{i=1}^c \pi^\star_i$ is a lower bound for
the optimum value of the dual \esp{} \ref{eq:dconv0}.
If $\ppp^\star \in
\{0,1\}^c$, $\ppp^\star$ is also
a globally optimal solution for \ref{eq:dconv0}. 
Otherwise we need a rounding scheme to map $\ppp^\star$ into a
feasible (suboptimal) solution for \ref{eq:dconv0}.
A simple deterministic rounding strategy is
the following.
\begin{itemize}
  \item[-] Step 1. Sort the edges in $\CCC^+$ according to $\ppp^\star$ in
    descending order.
  \item[-] Step 2. Pick edges from the sorted list until 
    $\log\det\LL(\ppp) \geq \delta + \tau_\text{init}$.
\end{itemize}
Theorem~\ref{th:cvx} allows us to interpret \ref{eq:dconv1} as
finding the optimal sampling probabilities $\ppp^\star$ that
minimizes the expected number of new edges such that the expected weighted
number of spanning trees is at least $\exp(\delta + \tau_\text{init})$; i.e.,
\begin{equation}
  \begin{aligned}
    & \underset{\ppp}{\text{minimize}}
    && {\mathbb{E}}_{{\mathcal{H} \sim
  \mathbb{G}(\GG_\bullet,\ppp_\bullet)}}
  \big[|\EE(\mathcal{H})|\big],\\
    & \text{subject to}
    & & {\mathbb{E}}_{{\mathcal{H} \sim
    \mathbb{G}(\GG_\bullet,\ppp_\bullet)}}
    \big[t_w(\mathcal{H})\big] \geq \exp(\delta+\tau_\text{init}), \\
    &&& 0\leq \pi_i \leq {1}, \, \forall i \in [s],
  \end{aligned}
  \label{eq:dconv1prime}
  \tag{D$_2^\prime$}
\end{equation}
in which $\mathbb{G}(\GG_\bullet,\ppp_\bullet)$ is defined in
Theorem~\ref{th:cvx}.
This narrative suggests a randomized rounding scheme in which $e_i \in \CCC^+$
is selected with probability $\pi^\star_i$. The expected number of selected
edges by this procedure is $\sum_{i=1}^{c} \pi^\star_i$.
\subsubsection{Certifying Near-Optimality}
\begin{corollary}
  Define $\zeta^\ast \triangleq 1/\gamma$ where $\gamma$ is the approximation
  factor given by Theorem~\ref{th:dualGreedy}. Let $k_{\normalfont\text{cvx}}$
  be the number of new edges selected by the deterministic rounding procedure
  described above.
  \begin{equation}
    \normalfont
    \max \bigg\{ \zeta^\ast k_\GRD, \Big\lceil \sum_{i=1}^{c} \pi^\star_i
  \Big\rceil \bigg\} \leq k_\OPT
    \leq \min \Big\{ k_\GRD,k_\text{cvx} \Big\}.
    \label{}
  \end{equation}
  \label{cor:dual}
\end{corollary}
As we did before for $k$-\esp{}$^+$, the lower bound provided by
Corollary~\ref{cor:dual} can be used to construct an upper bound for the gap between
$k_\OPT$ and any (feasible) suboptimal design with an objective value of
$k_\Acal$. Let $\mathcal{L} \triangleq \max \Big\{ \zeta^\ast k_\GRD,
\Big\lceil\sum_{i=1}^{c} \pi^\star_i \Big\rceil\Big\}$. $\mathcal{L}$ can be computed by running
Algorithm~\ref{alg:greedD} and solving the convex optimization problem
\ref{eq:dconv1}. Consequently, $k_\Acal - k_\OPT \leq k_\Acal - \mathcal{L}$ and
$k_\Acal/k_\OPT \leq k_\Acal/\mathcal{L}$.

\section{Conclusion}
\label{sec:conclusion}
We studied the problem of designing near-$t$-optimal graphs under several types
of constraints and formulations. Several new structures were revealed and
exploited to design efficient approximation algorithms.
In particular, we proved that the weighted number of
spanning trees in connected graphs can be posed as a monotone log-submodular
function of the edge set. Our approximation algorithms can find near-optimal
solutions with performance guarantees. They also provide \emph{a posteriori} near-optimality
certificates for arbitrary designs. Our results can be readily applied to a wide
verity of applications involving graph synthesis and graph sparsification
scenarios.

\bibliographystyle{plainnat}
\bibliography{../../rss16-trees/paper/graph}

\begin{thebibliography}{31}
\providecommand{\natexlab}[1]{#1}
\providecommand{\url}[1]{\texttt{#1}}
\expandafter\ifx\csname urlstyle\endcsname\relax
  \providecommand{\doi}[1]{doi: #1}\else
  \providecommand{\doi}{doi: \begingroup \urlstyle{rm}\Url}\fi

\bibitem[Bailey and Cameron(2009)]{bailey2009combinatorics}
Rosemary~A Bailey and Peter~J Cameron.
\newblock Combinatorics of optimal designs.
\newblock \emph{Surveys in Combinatorics}, 365:\penalty0 19--73, 2009.

\bibitem[Barooah and Hespanha(2007)]{Barooah2007}
Prabir Barooah and Joao~P Hespanha.
\newblock Estimation on graphs from relative measurements.
\newblock \emph{Control Systems, IEEE}, 27\penalty0 (4):\penalty0 57--74, 2007.

\bibitem[Bauer et~al.(1987)Bauer, Boesch, Suffel, and
  Van~Slyke]{bauer1987validity}
Douglas Bauer, Francis~T Boesch, Charles Suffel, and R~Van~Slyke.
\newblock On the validity of a reduction of reliable network design to a graph
  extremal problem.
\newblock \emph{Circuits and Systems, IEEE Transactions on}, 34\penalty0
  (12):\penalty0 1579--1581, 1987.

\bibitem[Boesch et~al.(2009)Boesch, Satyanarayana, and
  Suffel]{boesch2009survey}
Francis~T Boesch, Appajosyula Satyanarayana, and Charles~L Suffel.
\newblock A survey of some network reliability analysis and synthesis results.
\newblock \emph{Networks}, 54\penalty0 (2):\penalty0 99--107, 2009.

\bibitem[Boyd and Vandenberghe(2004)]{Boyd2004}
Stephen Boyd and Lieven Vandenberghe.
\newblock \emph{Convex optimization}.
\newblock Cambridge university press, 2004.

\bibitem[Cheng(1981)]{cheng1981maximizing}
Ching-Shui Cheng.
\newblock Maximizing the total number of spanning trees in a graph: two related
  problems in graph theory and optimum design theory.
\newblock \emph{Journal of Combinatorial Theory, Series B}, 31\penalty0
  (2):\penalty0 240--248, 1981.

\bibitem[Fisher et~al.(1978)Fisher, Nemhauser, and Wolsey]{fisher1978analysis}
Marshall~L Fisher, George~L Nemhauser, and Laurence~A Wolsey.
\newblock \emph{An analysis of approximations for maximizing submodular set
  functions—II}.
\newblock Springer, 1978.

\bibitem[Gaffke(1982)]{gaffke1982d}
N~Gaffke.
\newblock D-optimal block designs with at most six varieties.
\newblock \emph{Journal of Statistical Planning and Inference}, 6\penalty0
  (2):\penalty0 183--200, 1982.

\bibitem[Ghosh et~al.(2008)Ghosh, Boyd, and Saberi]{ghosh2008minimizing}
Arpita Ghosh, Stephen Boyd, and Amin Saberi.
\newblock Minimizing effective resistance of a graph.
\newblock \emph{SIAM review}, 50\penalty0 (1):\penalty0 37--66, 2008.

\bibitem[Godsil and Royle(2001)]{Godsil2001}
Chris Godsil and Gordon Royle.
\newblock \emph{Algebraic graph theory}.
\newblock Graduate Texts in Mathematics Series. Springer London, Limited, 2001.
\newblock ISBN 9780387952413.

\bibitem[Grant and Boyd(2008)]{gb08}
Michael Grant and Stephen Boyd.
\newblock Graph implementations for nonsmooth convex programs.
\newblock In V.~Blondel, S.~Boyd, and H.~Kimura, editors, \emph{Recent Advances
  in Learning and Control}, Lecture Notes in Control and Information Sciences,
  pages 95--110. Springer-Verlag Limited, 2008.
\newblock \url{http://stanford.edu/~boyd/graph_dcp.html}.

\bibitem[Grant and Boyd(2014)]{cvx}
Michael Grant and Stephen Boyd.
\newblock {CVX}: Matlab software for disciplined convex programming, version
  2.1.
\newblock \url{http://cvxr.com/cvx}, March 2014.

\bibitem[Hochbaum(1996)]{hochbaum1996approximation}
Dorit~S Hochbaum.
\newblock \emph{Approximation algorithms for NP-hard problems}.
\newblock PWS Publishing Co., 1996.

\bibitem[Joshi and Boyd(2009)]{Joshi2009}
Siddharth Joshi and Stephen Boyd.
\newblock Sensor selection via convex optimization.
\newblock \emph{Signal Processing, IEEE Transactions on}, 57\penalty0
  (2):\penalty0 451--462, 2009.

\bibitem[Kelmans(1996)]{kelmans1996graphs}
Alexander~K Kelmans.
\newblock On graphs with the maximum number of spanning trees.
\newblock \emph{Random Structures \& Algorithms}, 9\penalty0 (1-2):\penalty0
  177--192, 1996.

\bibitem[Kelmans and Kimelfeld(1983)]{kelmans1983multiplicative}
Alexander~K Kelmans and BN~Kimelfeld.
\newblock Multiplicative submodularity of a matrix's principal minor as a
  function of the set of its rows and some combinatorial applications.
\newblock \emph{Discrete Mathematics}, 44\penalty0 (1):\penalty0 113--116,
  1983.

\bibitem[Krause and Golovin(2012)]{krause2012submodular}
Andreas Krause and Daniel Golovin.
\newblock Submodular function maximization.
\newblock \emph{Tractability: Practical Approaches to Hard Problems},
  3:\penalty0 19, 2012.

\bibitem[L\"{o}fberg(2004)]{YALMIP}
Johan L\"{o}fberg.
\newblock Yalmip : A toolbox for modeling and optimization in {MATLAB}.
\newblock In \emph{Proceedings of the CACSD Conference}, Taipei, Taiwan, 2004.
\newblock URL \url{http://users.isy.liu.se/johanl/yalmip}.

\bibitem[Lov{\'a}sz(1993)]{lovasz1993random}
L{\'a}szl{\'o} Lov{\'a}sz.
\newblock Random walks on graphs: A survey.
\newblock \emph{Combinatorics, Paul Erd\H{o}s is eighty}, 2\penalty0
  (1):\penalty0 1--46, 1993.

\bibitem[Mesbahi and Egerstedt(2010)]{Mesbahi2010}
Mehran Mesbahi and Magnus Egerstedt.
\newblock \emph{Graph theoretic methods in multiagent networks}.
\newblock Princeton University Press, 2010.

\bibitem[Meyer(2000)]{Meyer2000}
Carl~D Meyer.
\newblock \emph{Matrix analysis and applied linear algebra}.
\newblock {SIAM}, 2000.

\bibitem[Myrvold(1996)]{myrvold1996reliable}
Wendy Myrvold.
\newblock Reliable network synthesis: Some recent developments.
\newblock In \emph{Proceedings of International Conference on Graph Theory,
  Combinatorics, Algorithms, and Applications}, 1996.

\bibitem[Nemhauser et~al.(1978)Nemhauser, Wolsey, and
  Fisher]{nemhauser1978analysis}
George~L Nemhauser, Laurence~A Wolsey, and Marshall~L Fisher.
\newblock An analysis of approximations for maximizing submodular set functions
  {- I}.
\newblock \emph{Mathematical Programming}, 14\penalty0 (1):\penalty0 265--294,
  1978.

\bibitem[Oxley(2006)]{oxley2006matroid}
James~G Oxley.
\newblock \emph{Matroid theory}, volume~3.
\newblock Oxford university press, 2006.

\bibitem[Petingi and Rodriguez(2002)]{petingi2002new}
Louis Petingi and Jose Rodriguez.
\newblock A new technique for the characterization of graphs with a maximum
  number of spanning trees.
\newblock \emph{Discrete mathematics}, 244\penalty0 (1):\penalty0 351--373,
  2002.

\bibitem[Pukelsheim(1993)]{Pukelsheim1993}
Friedrich Pukelsheim.
\newblock \emph{Optimal design of experiments}, volume~50.
\newblock SIAM, 1993.

\bibitem[Shier(1974)]{shier1974maximizing}
DR~Shier.
\newblock Maximizing the number of spanning trees in a graph with n nodes and m
  edges.
\newblock \emph{Journal Research National Bureau of Standards, Section B},
  78:\penalty0 193--196, 1974.

\bibitem[T{\"u}t{\"u}nc{\"u} et~al.(2003)T{\"u}t{\"u}nc{\"u}, Toh, and
  Todd]{tutuncu2003solving}
Reha~H T{\"u}t{\"u}nc{\"u}, Kim~C Toh, and Michael~J Todd.
\newblock Solving semidefinite-quadratic-linear programs using sdpt3.
\newblock \emph{Mathematical programming}, 95\penalty0 (2):\penalty0 189--217,
  2003.

\bibitem[Vandenberghe et~al.(1998)Vandenberghe, Boyd, and
  Wu]{vandenberghe1998determinant}
Lieven Vandenberghe, Stephen Boyd, and Shao-Po Wu.
\newblock Determinant maximization with linear matrix inequality constraints.
\newblock \emph{SIAM journal on matrix analysis and applications}, 19\penalty0
  (2):\penalty0 499--533, 1998.

\bibitem[Weichenberg et~al.(2004)Weichenberg, Chan, and
  M{\'e}dard]{weichenberg2004high}
Guy Weichenberg, Vincent~WS Chan, and Muriel M{\'e}dard.
\newblock High-reliability topological architectures for networks under stress.
\newblock \emph{Selected Areas in Communications, IEEE Journal on}, 22\penalty0
  (9):\penalty0 1830--1845, 2004.

\bibitem[Wolsey(1982)]{wolsey1982analysis}
Laurence~A Wolsey.
\newblock An analysis of the greedy algorithm for the submodular set covering
  problem.
\newblock \emph{Combinatorica}, 2\penalty0 (4):\penalty0 385--393, 1982.

\end{thebibliography}
\clearpage
\begin{appendix}
\section{Proofs}
\label{app:lemmas}
\begin{lemma}
  For any $\MM \in \mathbb{S}^n_{>0}$ and $\NN \in \mathbb{S}^{n}_{> 0}$, $\MM \succeq
  \NN$ iff
  $\NN^{-1} \succeq \MM^{-1}$.
  \label{th:inverseOp}
\end{lemma}
\begin{proof}
  Due to symmetry it suffices to prove that $\MM \succeq \NN \Rightarrow
  \NN^{-1} \succeq \MM^{-1}$. Multiplying both sides of $\MM \succeq \NN$ by
  $\NN^{-\frac{1}{2}}$ from left and right results in
  $\NN^{-\frac{1}{2}}\MM\NN^{-\frac{1}{2}} - \II \succeq \zero$. Therefore the
  eigenvalues of $\NN^{-\frac{1}{2}}\MM\NN^{-\frac{1}{2}}$, which are the same
  as the eigenvalues of
  $\MM^{\frac{1}{2}}\NN^{-1}\MM^{\frac{1}{2}}$,\footnote{Recall that $\MM\NN$
  and $\NN\MM$ have the same spectrum.} are at least $1$. Therefore
  $\MM^{\frac{1}{2}}\NN^{-1}\MM^{\frac{1}{2}} - \II \succeq \zero$. Multiplying
  both sides by $\MM^{-\frac{1}{2}}$ from left and right proves the lemma.
\end{proof}

\begin{lemma}[Matrix Determinant Lemma]
  For any non-singular $\MM \in \mathbb{R}^{n \times n}$ and $\cc,\dd \in
  \mathbb{R}^{n}$,
  \begin{equation}
    \det(\MM + \cc \dd^{\top}) = (1+\dd^{\top} \MM^{-1} \cc) \det\MM.
    \label{}
  \end{equation}
  \label{th:rank1}
\end{lemma}
\begin{proof}
  See e.g., \cite{Meyer2000}.
\end{proof}

\begin{lemma}
  Let $\GG_1$ be a spanning subgraph of $\GG_2$. For any $w:\EE(K) \to
  \mathbb{R}_{\geq 0}$, $\LL^{w}_{\GG_2} \succeq
  \LL^{w}_{\GG_1}$ in which $\LL^{w}_{\GG}$ is the reduced weighted Laplacian
  matrix of $\GG$ when its edges are weighted by $w$.
  \label{th:subgraph}
\end{lemma}
\begin{proof}
  From the definition of the reduced weighted Laplacian matrix we have,
  \begin{equation}
    \LL^{w}_{\GG_2} - \LL^{w}_{\GG_1} = \sum_{\mathclap{\{u,v\} \in \EE(\GG_2)
    \setminus \EE(\GG_1)}}
    w_{uv} \,\,\aaa_{uv}\aaa_{uv}^\top \succeq \zero.
    \label{}
  \end{equation}
\end{proof}
\begin{proof}[Proof of Theorem~\ref{th:expected}]
  Define the following indicator function,
  \begin{equation}
    \mathbbm{1}_{\mathpzc{T}_{\GG}}(\TT) \triangleq
    \begin{cases}
      1 & {\TT \in \mathpzc{T}_{\GG}}, \\
      0 & {\TT \notin \mathpzc{T}_{\GG},}
    \end{cases}
    \label{eq:indicator}
  \end{equation}
  in which $\mathpzc{T}_{\GG}$ denotes the set of spanning trees of $\GG$.
  Now note that,
  \begin{align}
    {\mathbb{E}}_{{{\GG \sim \mathbb{G}(\GG^\circ,\pp)}}}\big[t_w(\GG)\big]
    & 
    =  \mathbb{E}_{{{\GG \sim
      \mathbb{G}(\GG^\circ,\pp)}}}
      \Big[\,\sum_{\mathclap{\TT \in
       \mathpzc{T}_{\GGc}}}
      \mathbbm{1}_{\mathpzc{T}_{\GG}}(\TT)\mathbb{V}_w(T) \Big] \\
    &  = \sum_{\mathclap{\TT
      \in \mathpzc{T}_{\GGc}}}
      \mathbb{E}_{{{{\GG \sim \mathbb{G}(\GG^\circ,\pp)}}}}
      \Big[\mathbbm{1}_{\mathpzc{T}_{\GG}}(\TT) \mathbb{V}_w({\TT})\Big] \\
    & =  \sum_{\mathclap{\TT
      \in \mathpzc{T}_{\GGc}}}
      \mathbb{P}\Big[\TT \in \mathpzc{T}_{\GG}\Big] \mathbb{V}_w(T) \\
    & =  \sum_{\mathclap{\TT
      \in \mathpzc{T}_{\GGc}}}
      \mathbb{V}_p(\TT) \mathbb{V}_{w}(\TT) \\
      & =\sum_{\mathclap{\TT
	\in \mathpzc{T}_{\GGc}}}
	\mathbb{V}_{w_p}(\TT) \\
	& = t_{w_p}(\GGc).
    \label{}
  \end{align}
  Here we have used the fact the $\mathbb{P}[T \in \mathpzc{T}_\GG]$ is equal to
  the probability of existence of every edge of $T$ in $\GG$, which is equal to
  $\mathbb{V}_p(\TT)$.
\end{proof}
\begin{proof}[Proof of Lemma~\ref{th:add}]
  Note that $\LL_{\GG^+} =
  \LL_{\GG} +
  w_{uv}\,\aaa_{uv}\aaa_{uv}^{\top}$. Taking the determinant, applying
  Lemma~\ref{th:rank1} and taking the $\log$ concludes the proof.
\end{proof}
\begin{proof}[Proof of Lemma~\ref{th:del}]
  The proof is similar to the proof of Lemma~\ref{th:add}.
\end{proof}
\begin{proof}[Proof of Theorem~\ref{th:treeSupermodular}]
  First recall that $\mathbb{V}_w(T)$ is positive for any $T$ by definition.
  \begin{enumerate}
    \item Normalized:
      $\tree(\varnothing) = 0$ by definition.
    \item Monotone: Let $G \triangleq (\VV,\EE \cup \{e\})$. Denote by
      $\mathpzc{T}_\GG^e$ the set of spanning trees of $\GG$ that contain
      $e$.
      \begin{align}
	\tree(\EE \cup \{e\}) & = \sum_{\mathclap{T \in \mathpzc{T}_{\GG}}} \mathbb{V}_w(T)
	= \sum_{\mathclap{T \in \mathpzc{T}_{\GG}^e}} 
	\mathbb{V}_w(T) + 
	\sum_{\mathclap{T \notin \mathpzc{T}_{\GG}^e}} \mathbb{V}_w(T) \\
	&= \sum_{\mathclap{T \in \mathpzc{T}_{\GG}^e}} 
	\mathbb{V}_w(T) +
	\tree(\EE) \geq \tree(\EE).
	\label{eq:sumTree}
      \end{align}
    \item Supermodular:
  $\tree$ is supermodular iff for all
  $\EE_1 \subseteq \EE_2 \subseteq \EE(K_n)$ and all $e\in \EE(K_n) \setminus
  \EE_2$,
  \begin{equation}
    \tree(\EE_2 \cup \{e\}) - \tree(\EE_2) \geq \tree(\EE_1 \cup
    \{e\}) - \tree(\EE_1).
    \label{eq:treesupproof}
  \end{equation}
  Define $\GG_1 \triangleq (\VV,\EE_1)$ and $\GG_2 \triangleq (\VV,\EE_2)$. As we showed in \eqref{eq:sumTree},
  \begin{align}
    \tree(\EE_1 \cup \{e\}) - \tree(\EE_1) & =  \sum_{\mathclap{T \in
      \mathpzc{T}_{\GG_1}^e}} 
	\mathbb{V}_w(T), \\
    \tree(\EE_2 \cup \{e\}) - \tree(\EE_2) & =  
    \sum_{\mathclap{T \in
      \mathpzc{T}_{\GG_2}^e}} 
	\mathbb{V}_w(T).
    \label{}
  \end{align}
  Therefore we need to show that $\sum_{{T \in \mathpzc{T}_{\GG_2}^e}}
  \mathbb{V}_w(T) \geq  \sum_{{T \in \mathpzc{T}_{\GG_1}^e}}
  \mathbb{V}_w(T)$. This inequality holds since $\mathpzc{T}_{\GG_1}^e \subseteq
  \mathpzc{T}_{\GG_2}^e$.
  \end{enumerate}
  
\end{proof}
\begin{proof}[Proof of Theorem~\ref{th:logTG}]
  $\,$\\
  \vspace{-0.5cm}
  \begin{enumerate}
    \item Normalized: By definition $\logTG(\varnothing) = \logtree(\EEi) - \logtree(\EEi) =
      0$.
    \item Monotone: We need to show that $\logTG(\EE \cup \{e\}) \geq
      \logTG(\EE)$. This is equivalent to showing that,
      \begin{align}
	\logtree(\EEi \cup \EE \cup \{e\}) \geq \logtree(\EEi \cup \EE).
	\label{eq:logtreeMon}
      \end{align}
      Now note that $(\VV,\EEi \cup \EE)$ is connected since
      $(\VV,\EEi)$ was assumed to be connected. Therefore we can apply 
      Lemma~\ref{th:add} on the LHS of \eqref{eq:logtreeMon}; i.e.,
      \begin{equation}
      \logtree(\EEi \cup \EE \cup \{e\}) = \logtree(\EEi \cup \EE) +
      \log(1+w_e\Delta_e).
	\label{}
      \end{equation}
      Therefore it sufficies to show that $\log(1+w_e\Delta_e)$ is non-negative.
      Since $(\VV,\EEi)$ is connected, $\LL$ is positive definite. Consequently $w_e \Delta_e  = w_e \aaa_e^\top
      \LL^{-1} \aaa_e > 0$ and hence $\log(1+w_e\Delta_e) > 0$.
    \item Submodular:
  $\logTG$ is submodular iff for all
  $\EE_1 \subseteq \EE_2 \subseteq \EE(K_n)$ and all $e\in \EE(K_n) \setminus
  \EE_2$,
  \begin{equation}
    \logTG(\EE_1 \cup \{e\}) - \logTG(\EE_1) \geq \logTG(\EE_2 \cup
    \{e\}) - \logTG(\EE_2).
    \label{eq:logTGSubmodular}
  \end{equation}
  After canceling $\logtree(\EEi)$ we need to show that,
  \begin{equation}
    \logtree(\EE_1 \cup \EEi \cup \{e\}) - \logtree(\EE_1 \cup \EEi) \geq
    \logtree(\EE_2 \cup \EEi \cup
    \{e\}) - \logtree(\EE_2 \cup \EEi).
    \label{eq:logTGSubmodular2}
  \end{equation}
  If $e \in \EEi$, both sides of \eqref{eq:logTGSubmodular2} become zero. Hence we
  can safely assume that $e \notin \EEi$.
  To shorten our notation let us define $\EE^*_i \triangleq \EE_i \cup \EEi$
  for $i=1,2$. Therefore \eqref{eq:logTGSubmodular2} can be rewritten as,
  \begin{equation}
    \logtree(\EE_1^* \cup \{e\}) - \logtree(\EE_1^*) \geq
    \logtree(\EE_2^* \cup
    \{e\}) - \logtree(\EE_2^* ).
    \label{eq:logTGSubmodular3}
  \end{equation}
  Recall that by assumption $(\VV,\EEi)$ is connected. Thus $(\VV,\EE_i^*)$ is
  connected for $i = 1,2$, and we can
  apply Lemma~\ref{th:add} on both sides of \eqref{eq:logTGSubmodular3}. After
  doing so we have to show that
  \begin{align} 
    \log(1+w_e\Delta_e^{\GG_1}) &\geq \log(1+w_e\Delta_e^{\GG_2})
    \label{eq:preCond}
  \end{align}
  where $\GG_i \triangleq (\VV,\EE_i \cup \EEi,w)$ for $i=1,2$. It is easy to
  see that \eqref{eq:preCond} holds iff $\Delta_e^{\GG_1} \geq
  \Delta_e^{\GG_2}$.
  Now note that
  \begin{align}
    \Delta_e^{\GG_1} - \Delta_e^{\GG_2} = \aaa_e^\top (\LL_{\GG_1}^{-1} -
    \LL_{\GG_2}^{-1}) \, \aaa_{e} \geq 0
    \label{}
  \end{align}
  since $\LL_{\GG_2} \succeq \LL_{\GG_1}$ ($\GG_1$ is a spanning subgraph of
  $\GG_2$), and therefore according to Lemma~\ref{th:inverseOp}
  $\LL_{\GG_1}^{-1} \succeq \LL_{\GG_2}^{-1}$.
  \end{enumerate}
\end{proof}
\begin{proof}[Proof of Theorem~\ref{th:cvx}]
  First note that \eqref{eq:i} directly follows from Theorem~\ref{th:expected}
  since $\LL(\ppp)$ is the reduced weighted Laplacian matrix of $\GG_\bullet$
  after scaling its edge weights by the sampling probabilities $\pi_1,\dots,\pi_s$.
  To prove \eqref{eq:ii} consider the following indicator function,
  \begin{equation}
    \mathbbm{1}_{\EE(\mathcal{H})}(e) = 
    \begin{cases}
      1 & e \in \EE(\mathcal{H}), \\
      0 & e \notin \EE(\mathcal{H}).
    \end{cases}
    \label{}
  \end{equation}
  Now note that $\mathbbm{1}_{\EE(\mathcal{H})}(e_i) \sim
  \operatorname{Bern}(\pi_i)$ for $i =1,\dots,s$. Therefore,
  \begin{align}
    {\mathbb{E}}_{{\mathcal{H} \sim
  \mathbb{G}(\GG_\bullet,\ppp_\bullet)}}
  \big[|\EE(\mathcal{H})|\big] & 
   = 
    {\mathbb{E}}_{{\mathcal{H} \sim
  \mathbb{G}(\GG_\bullet,\ppp_\bullet)}}
  \Big[ \sum_{i=1}^{s} \mathbbm{1}_{\EE(\mathcal{H})}(e_i) \Big] \\
  & =
  \sum_{i=1}^{s}{\mathbb{E}}_{{\mathcal{H} \sim
  \mathbb{G}(\GG_\bullet,\ppp_\bullet)}}
  \Big[\mathbbm{1}_{\EE(\mathcal{H})}(e_i) \Big]  \\
  & = \sum_{i=1}^{s} \pi_i \\
  & = \sum_{\mathclap{e_i \in \CCC^+}}\pi_i + \sum_{\mathclap{e_j \in \EEi}} 1 \\
  & = \|\ppp\|_1 + |\EEi|.
    \label{}
  \end{align}
\end{proof}
\begin{proof}[Proof of Theorem~\ref{th:Cher}]
  This theoreom is a direct application of Chernoff bounds for Poisson
  trials of independently sampling edges from $\CCC^+$ with probabilities
  specified by $\ppp^\star$.
\end{proof}

\subsection*{Generalizing Theorems~\ref{th:expected} and \ref{th:cvx}}
\label{app:genCvx}
The following theorem generalizes Theorem~\ref{th:expected} (and, consequently,
Theorem~\ref{th:cvx}).
Theorem~\ref{th:detGeneral} provides a 
similar interpretation for the convex relaxation approach designed by
\citet{Joshi2009} for the sensor selection problem with
linear measurement models.
\begin{theorem}
  Let $\{(\yy_i,\zz_i)\}_{i=1}^{m}$ be a collection of $m$ pairs of vectors in $\mathbb{R}^n$
  such that $m \geq n$. Furthermore, let $s_1,\dots,s_m$ be a collection of
  $m$ independent random variables such that  $s_i \sim \mathrm{Bern}(p_i)$ for
  some $p_i \in [0,1]$. Then we have,
  \begin{equation}
    \underset{\underset{\forall i \in [m]}{s_i \sim \mathrm{Bern}(p_i)}}{\mathbb{E}}
    \Big[\det
    \left(\sum_{i=1}^m s_i\yy_i \zz_i^\top \right)\Big] =
      \det
      \left(\sum_{i=1}^m p_i\yy_i \zz_i^\top \right).
    \label{}
  \end{equation}
  \label{th:detGeneral}
\end{theorem}
\begin{proof}[Proof] 
  Let $\mathcal{S}_n \triangleq \binom{[m]}{n}$ be the set of all $n$-subsets
  of $[m]$.
 According to the Cauchy-Binet (C-B) formula we have 
 \begin{align}
   \underset{\underset{\forall i \in [m]}{s_i \sim \mathrm{Bern}(p_i)}}{\mathbb{E}}
   \Big[ \det
   \left(\sum_{i=1}^m s_i\yy_i \zz_i^\top \right)\Big] &
      \overset{\text{C-B}}{=}
      \underset{\underset{\forall i \in [m]}{s_i \sim
	\mathrm{Bern}(p_i)}}{\mathbb{E}} \Big[
	\sum_{\mathcal{Q} \in \mathcal{S}_n}\det
      \left(\sum_{i \in \mathcal{Q}} s_i\yy_i \zz_i^\top \right)\Big]  \\
      & \overset{\phantom{\text{C-B}}}{=}
	\sum_{\mathcal{Q} \in \mathcal{S}_n} 
	\underset{\underset{\forall i \in [m]}{s_i \sim
	  \mathrm{Bern}(p_i)}}{\mathbb{E}} \Big[ \det
	  \left(\sum_{i \in \mathcal{Q}} s_i\yy_i \zz_i^\top \right)\Big].
   \label{eq:firstExpansion}
 \end{align}
 Now note that $|\cQQ| = n$ and $\rank(\yy_i \zz_i^\top) = 1$. Therefore $\det \left( \sum_{i\in \cQQ} s_i
 \yy_i \zz_i^\top \right)$ is non-zero iff $s_i = 1$ for all $i \in \cQQ$.
 Thus for every $\cQQ \in \mathcal{S}_n$,
 \begin{equation}
   \det \left( \sum_{i\in \cQQ} s_i
 \yy_i \zz_i^\top \right) =
 \begin{cases}
   d_{\cQQ} \triangleq \det \left( \sum_{i \in \cQQ} \yy_i \zz_i^\top \right) &
   \text{with probability $p_{\cQQ} \triangleq \prod_{i \in \cQQ} p_i$}, \\
   0 & \text{with probability $1 - p_{\cQQ}$}.
 \end{cases}
   \label{}
 \end{equation}
 Taking the expectation yields
 \begin{equation}
   \underset{\underset{\forall i \in [m]}{s_i \sim \mathrm{Bern}(p_i)}}{\mathbb{E}}
   \Big[
   \det \left( \sum_{i\in \cQQ} s_i
 \yy_i \zz_i^\top \right) \Big] = p_{\cQQ} d_{\cQQ}.
   \label{eq:pqdq}
 \end{equation}
 Replacing \eqref{eq:pqdq} in \eqref{eq:firstExpansion} results in
 \begin{align}
   \underset{\underset{\forall i \in [m]}{s_i \sim \mathrm{Bern}(p_i)}}{\mathbb{E}}
   \Big[ \det
   \left(\sum_{i=1}^m s_i\yy_i \zz_i^\top \right)\Big] &
      =
	\sum_{\mathcal{Q} \in \mathcal{S}_n} 
	p_{\cQQ} d_{\cQQ} =  \sum_{\cQQ \in \mathcal{S}} \det \left( \sum_{i \in \cQQ} p_i \yy_i
 \zz_i^\top \right).
   \label{eq:secondExpansion}
 \end{align}
 Noting that the RHS in \eqref{eq:secondExpansion} is the
 Cauchy-Binet expansion of $\det \left(\sum_{i=1}^m p_i\yy_i \zz_i^\top
 \right)$ concludes the proof. 
\end{proof}

\end{appendix}

\end{document}